\documentclass[11pt,a4paper]{article}
\usepackage[T1]{fontenc}
\usepackage[utf8]{inputenc}
\usepackage{lmodern}
\usepackage[textwidth=6.5in,top=1in,bottom=1in,headheight=0pt,headsep=0pt]{geometry}
\usepackage{titling}
\usepackage{authblk}
\usepackage{amsmath,amssymb,amsthm,mathtools}
\usepackage{microtype}
\usepackage{enumitem}
\usepackage{booktabs}
\usepackage{needspace}
\usepackage[font=small,labelfont=bf,labelsep=period]{caption}
\usepackage{tikz}
\usetikzlibrary{arrows.meta,positioning,calc}
\usepackage[numbers]{natbib}
\usepackage{xcolor}
\usepackage{xurl}
\usepackage{hyperref}
\hypersetup{
 colorlinks=true,
 linkcolor=blue!50!black,
 citecolor=blue!50!black,
 urlcolor=blue!50!black,
 pdftitle={A Dichotomy for Planar Graph Homomorphisms with Nonnegative Weights},
 pdfauthor={Chenghua Liu and Boning Meng},
 pdfsubject={Research manuscript: exact counting, planar gadgets, and tensor structure}
}
\newtheorem{theorem}{Theorem}[section]
\newtheorem{lemma}[theorem]{Lemma}
\newtheorem{proposition}[theorem]{Proposition}
\newtheorem{corollary}[theorem]{Corollary}
\theoremstyle{definition}

\theoremstyle{remark}

\newtheorem*{remark*}{Remark}
\numberwithin{equation}{section}
\setlist{topsep=4pt,itemsep=2pt,parsep=0pt}
\allowdisplaybreaks[1]
\newcommand{\PlGH}[1]{\operatorname{Pl\text{-}GH}(#1)}

\title{A Dichotomy for Planar Graph Homomorphisms\\
with Nonnegative Weights}
\author[1]{Chenghua Liu}
\author[2]{Boning Meng}
\affil[1]{Institute of Software, Chinese Academy of Sciences, Beijing, China}
\affil[2]{University of Regensburg, Regensburg, Germany}
\affil[ ]{\href{mailto:liuch.russell@gmail.com}{\texttt{liuch.russell@gmail.com}},
\href{mailto:mengboning2013@gmail.com}{\texttt{mengboning2013@gmail.com}}}
\date{}

\begin{document}
\maketitle

\begin{abstract}
We prove a complete complexity dichotomy for planar graph homomorphism
counting with any fixed symmetric nonnegative matrix of arbitrary finite
order, giving an explicit criterion for tractability. We also
characterize exactly which fixed positive vertex weights preserve
tractability, with both classifications extending from algebraic weights
to fixed real weights in a prescribed exact representation. Our proof
hinges on an entropy-based continuation argument:
maximal logarithmic support identifies distance kernels as maximum-entropy
completions, extending their positive definiteness throughout the parameter
interval. This enables distance geometry to recover hidden product
coordinates even when planar gadgets cannot distinguish colors;
counting-hardness arguments then force the factors to be zero-field
Boolean Ising interactions.
The classification also yields complete tractability criteria for clock
models, coupled Ising systems, and planar contractions of stoquastic
imaginary-time kernels. All results have been formally verified in Lean~4.%
\footnote{The formalization is available at
\url{https://github.com/liuchliuch/planar-homomorphisms-lean}.}
\end{abstract}

\section{Introduction}\label{sec:introduction}

Counting graph colorings and evaluating spin partition functions are
two versions of the same problem: sum the weights of all assignments
of states to vertices, with interactions along edges. In combinatorics,
an interaction can prescribe which pairs of colors are allowed; in
statistical physics, it assigns a Boltzmann weight to each pair of
spins. The resulting partition function is both a counting invariant
and the normalizing constant for a Gibbs distribution. Understanding
when it can be computed efficiently is a common question across these
subjects.

Planarity makes this question particularly interesting. The
Fisher--Kasteleyn--Temperley (FKT) algorithm for the zero-field Ising model
\cite{kasteleyn1961,kasteleyn1963,temperleyFisher1961} is a striking
exception to the hardness of exact counting on general graphs
\cite{bulatovGrohe2005}. Independent Ising systems can be combined
to give tractable models with arbitrarily many states. Does every
tractable nonnegative interaction arise this way, apart from elementary
local summations? And which single-site fields preserve this advantage
of planarity? These questions ask for the extent of the Ising
phenomenon, beyond a particular model or number of states.

We study them in the language of weighted graph homomorphisms.
Each vertex receives one of $q$ states, also called colors, indexed
by $[q]=\{1,\ldots,q\}$. A fixed symmetric nonnegative matrix $M$
specifies the edge weight $M_{ij}$ when its endpoints receive states
$i,j$. The partition function is
\[
 Z_M(G)=\sum_{\sigma:V(G)\to[q]}
                  \prod_{uv\in E(G)}M_{\sigma(u),\sigma(v)}.
\]
When $M$ is a zero--one adjacency matrix, this counts homomorphisms
to its target graph. In a spin system, $M_{ij}=e^{-\beta E_{ij}}$
encodes the pair energy $E_{ij}$ at inverse temperature $\beta$;
zero entries allow hard exclusions. Write $\PlGH{M}$ for exact
evaluation on planar input graphs. The interaction and its number of
states are fixed; the size and shape of the input graph vary.

\subsection{Main results}

Our classification gives an affirmative answer to the first question:
\emph{all additional tractability supplied by planarity comes from
independent zero-field Ising coordinates}. This holds for every
nonnegative interaction, including singular matrices and matrices
with zero entries. The classification also determines exactly which fixed
positive vertex weights preserve tractability.

To state the result, write
\begin{equation}\label{eq:ising-factors}
 W(\rho)=\begin{pmatrix}1&\rho\\\rho&1\end{pmatrix},\qquad
 T(\boldsymbol\rho)=\bigotimes_{r=1}^dW(\rho_r),\qquad \rho_r>0.
\end{equation}
We allow $d=0$, with $T(\varnothing)=[1]$.
A color of $T$ is a bit string: each coordinate contributes an
independent two-state interaction. The tensor is \emph{nondegenerate}
when every $\rho_r\ne1$. Write $M\cong N$ for equality after a
simultaneous row and column permutation. The support graph of $M$
has vertex set $[q]$ and an edge $ij$, including a possible loop,
when $M_{ij}>0$. We call its connected components the
\emph{support components} of $M$.

\begin{theorem}[Nonnegative interactions]\label{thm:main}
Let $M=M^{\mathsf T}\in\mathbb R_{\ge0}^{q\times q}$ have fixed
algebraic entries, and write $M\cong\bigoplus_{\alpha=1}^s M_\alpha$
for its support components. Then $\PlGH{M}\in\mathsf{FP}$ if each
$M_\alpha$, after relabeling, has one of the forms
\begin{equation}\label{eq:tractable-forms}
 [0],\qquad
 (aa^{\mathsf T})\otimes T(\boldsymbol\rho),\qquad
 \begin{pmatrix}0&ab^{\mathsf T}\\ba^{\mathsf T}&0\end{pmatrix}
       \otimes T(\boldsymbol\rho),
\end{equation}
with
\[
 a\in\mathbb R_{>0}^{k},\quad b\in\mathbb R_{>0}^{\ell},\quad
 \boldsymbol\rho\in\mathbb R_{>0}^{d},\qquad k,\ell\ge1,\quad d\ge0.
\]
The parameters may differ between components. In every other case,
$\PlGH{M}$ is $\#\mathrm P$-hard under polynomial-time Turing reductions.
\end{theorem}

The algorithm rests on a general tensor-product identity: for any
interaction matrices $A,B$ and any graph $G$,
\[
 Z_{A\otimes B}(G)=Z_A(G)Z_B(G).
\]
A color of $A\otimes B$ is a pair of colors, and the two coordinate
colorings contribute independently to the sum. Thus each tensor
factor in \eqref{eq:tractable-forms} can be evaluated separately.
For the rank-one factor $aa^{\mathsf T}$, each vertex $v$ contributes
the independent sum $\sum_i a_i^{\deg(v)}$. The bipartite rank-two
factor gives zero on a nonbipartite input; on a connected bipartite
input, each of the two assignments of its parts to the target sides
gives such local sums, using $a$ and $b$. These factors are tractable
on arbitrary graphs \cite{bulatovGrohe2005}. Each $W(\rho_r)$ factor
is evaluated on planar graphs by the FKT algorithm, and the factor
values are multiplied. Zero singletons
contribute only on edgeless inputs. For each connected input component,
sum over target support components; then multiply over input
components. Large state spaces can therefore be tractable, but their
nontrivial interactions must have an underlying binary product structure.

For the second question, give state $i$ a fixed positive vertex weight
$w_i$, the same at every input vertex. Define
\[
 Z_{M,w}(G)
 =\sum_{\sigma:V(G)\to[q]}
   \left(\prod_{v\in V(G)}w_{\sigma(v)}\right)
   \left(\prod_{uv\in E(G)}M_{\sigma(u),\sigma(v)}\right).
\]
The corresponding planar evaluation problem is $\PlGH{M,w}$.
Physically, $w_i=e^{-\beta U_i}$ represents a single-site potential.
For a nondegenerate Ising tensor, the weighted problem is tractable
exactly when every Ising bit string has the same vertex weight.

If rows $i$ and $j$ of $M$ are identical, assigning $i$ or $j$ to a
vertex gives the same weight on each incident edge. We can therefore
replace these two states by one state with vertex weight $w_i+w_j$,
without changing the partition function. Theorem~\ref{thm:weighted}
gives the complete weighted classification for the forms in
\eqref{eq:tractable-forms}, accounting for these sums and the
elementary local factors.

The algebraic formulation uses exact number-field encodings.
Theorem~\ref{thm:real} extends both classifications to arbitrary fixed
real weights in the field representation of
\cite[Appendix~A]{cmy2026}, so they apply to Boltzmann parameters
without an algebraicity restriction.

\subsection{Applications}

The structural criterion turns questions about individual models into
short tests for independent Ising coordinates and compatible fields.
Section~\ref{sec:applications} develops the following consequences
for exact evaluation on planar graphs.

\begin{itemize}[leftmargin=*,itemsep=5pt]
\item \textbf{Graph homomorphisms.}
For unweighted targets, planarity does not enlarge the tractable
class of Dyer and Greenhill \cite{dyerGreenhill2000,dyerGreenhill2004}:
every other target remains $\#\mathrm P$-complete even on planar inputs.
Thus the additional tractability supplied by planarity is a weighted
phenomenon, arising from independent zero-field Ising factors.

\item \textbf{Planar bipartite graphs.}
We obtain a complete dichotomy also when inputs are required to be
bipartite, allowing different fixed vertex weights on the two
prescribed parts. This restriction strictly enlarges the tractable
class: some interactions that are hard on general planar graphs
become tractable (Section~\ref{sec:bipartite-inputs}).

\item \textbf{Clock and spin-one models.}
In the clock model, spins occupy $q$ equally spaced angles, and their
pair weight is the exponential of a coupling times the cosine of
their angular difference \cite{kim2017}. For every nonzero real
coupling, planar evaluation is tractable exactly for $q=2,4$;
positive vertex weights must also be constant
(Corollary~\ref{cor:clock}). The four-state exception consists of two
independent Ising bits. In contrast, the spin-one Blume--Capel model
\cite{buteraPernici2018} is hard at every nonzero pair coupling,
for all finite single-site potentials.

\item \textbf{Coupled Ising systems.}
A state can contain several Ising spins, with interactions coupling
products of these spins. With unit vertex weights, tractability is equivalent to linear
independence of the binary labels of the active interactions
(Corollary~\ref{cor:characters}): these products must be able to
serve as independent coordinates. For Ashkin--Teller systems
\cite{grestWidom1981} with only four-spin couplings between layers,
this says that the fixed graph describing which layers are coupled
must be a forest. The weighted criterion specifies which single-site
potentials preserve these decompositions.

\item \textbf{Quantum imaginary-time kernels.}
For a fixed real symmetric Hamiltonian $h$ with nonpositive
off-diagonal entries and a connected transition graph, the
imaginary-time kernel $e^{-\tau h}$ is positive for every $\tau>0$.
With unit vertex weights, its contraction on planar graphs is tractable
exactly when, after relabeling states and adding a scalar,
$h$ is a sum of independent negative bit-flip terms
(Corollary~\ref{cor:quantum-kernel}). For the quantum clock hopping
term \cite{vishnu2025}, this again singles out $q=2,4$.
\end{itemize}

The classifications persist on graphs of any fixed bounded orientable
genus, where the extension of the FKT algorithm uses a bounded
number of Pfaffians
\cite{pham2020}. The same structural boundary therefore describes
the algorithmic benefit of bounded topology.

\subsection{Related work}

Complexity classifications organize this landscape by asking which
local interactions make counting easy. For graph homomorphisms on
general graphs, the theory progressed from unweighted targets
\cite{dyerGreenhill2000,dyerGreenhill2004} to nonnegative weights
\cite{bulatovGrohe2005}, then to real and complex algebraic weights
\cite{goldbergGroheJerrumThurley2010,caiChenLu2013}.
The counting constraint satisfaction framework allows relations of
any arity, with corresponding arbitrary-domain classifications
\cite{bulatov2008,caiChenLu2016,caiChen2017}.
These results make tractability a structural property of the
interactions, rather than a collection of isolated algorithms.

For planar inputs, Valiant's holographic algorithms
\cite{valiant2008} use the FKT algorithm within a broader framework.
The planar Boolean counting CSP classifications
\cite{caiLuXia2017,guoWilliams2020,caiFu2016} identify the additional
tractable cases through matchgates, which encode weighted perfect
matchings. General planar Holant problems, with local constraints on
edge assignments, admit further tractable families
\cite{caiFuGuoWilliams2022}. The present classification identifies
the full extent of this phenomenon for nonnegative graph
homomorphisms with arbitrarily many colors.

Within this problem, Cai and Maran classified real three-state
matrices and nonnegative full-rank four-state matrices
\cite{caiMaran2023,caiMaran2024}. Cai, Maran and Young
\cite{cmy2026} obtained arbitrary-domain results by diagonal
separation: planar edge gadgets distinguish colors through their
diagonal entries. They related obstructions to separation to quantum
automorphisms. Highly symmetric interactions require a different
route. For example, $W(1/2)\otimes W(1/3)$ has rank four and admits
independent coordinate flips acting transitively on its colors.
Every edge gadget consequently has constant diagonal, yet the
partition function is the product of two Ising partition functions.
Our proof recovers hidden coordinates while retaining this symmetry.

\subsection{Proof overview}

The central difficulty is to recover independent coordinates without
assuming that gadgets can distinguish the colors. Our starting invariant
is the off-diagonal support of the \emph{ordinary matrix logarithm}.
For a positive definite Ising tensor, this support is a Boolean cube
\eqref{eq:ising-log}.

We begin with a positive definite interaction available through planar
reductions and choose one whose logarithmic support graph has as many
edges as possible. This maximality turns the freedom to construct
further gadgets into a restriction: no construction can preserve its edges and add a new
one. Interpolation then supplies matrices whose entries are powers
of a parameter, with exponent equal to distance in this graph.
The interpolation uses finitely many exact oracle queries.
Near zero, maximality forces the logarithm of each distance kernel
to have no off-diagonal entries outside the chosen graph. Further planar gadgets
restrict shortest paths and common neighbors. In particular, every
distance-two pair has exactly two common neighbors, which are
nonadjacent.

The bridge from these algebraic constraints to global structure is
\emph{matrix entropy}. The distance matrices are initially positive
definite only near one endpoint of the parameter range. By viewing
them as unique maximum-entropy completions of prescribed entries,
we extend positive definiteness throughout the interval.
The completion fixes the diagonal and the entries on graph edges,
leaving the other entries free. Its stationarity equations set the
corresponding entries of the logarithm to zero, precisely the
condition supplied by maximality. Uniqueness and analytic dependence
on the parameter then identify the completion with the distance
kernel throughout the interval.
This links
the counting problem to Gibbs manifolds and logarithmic sparsity
\cite{pavlov2024,pavlovSturmfelsTelen2024}. After normalization, the
same completion is quantum maximum-entropy inference: among density
matrices with specified expectation values, it selects the state of
largest von Neumann entropy \cite{watrous2018}.

Positive definiteness now brings in the classical connection between
distance kernels and Euclidean embeddings \cite{schoenberg1938}.
Graph distances become squared Euclidean distances; local squares
become parallelograms. Together with maximality, this forces a
Cartesian product of complete graphs.
The parallelogram rule lets us transport the cliques through one
vertex consistently across the graph. Their spans are orthogonal,
so they define separate coordinates; connectedness and this
transport show that every combination of coordinates occurs.
Known hardness results rule
out coordinates with three or more states and unequal-diagonal
two-state factors, leaving independent zero-field Ising coordinates.
To apply these hardness results, interpolation recovers individual
factors from the available partition functions. Once larger cliques
are excluded, maximality forces the logarithms of the remaining
positive definite interactions to separate into one-bit terms.
Exponentiation gives tensor decompositions in common coordinates,
which are also needed to compare auxiliary interactions later.
Thus entropy supplies a global positivity
principle, and geometry turns it into the product structure required
by the counting algorithm.

We next remove the positive-definiteness assumption. First, we classify
the possible zero patterns: unless hardness already follows, each
support component is entrywise positive, has complete bipartite
support, or is a zero singleton. After separating the elementary
local factors, we treat the two nontrivial cases. For an entrywise
positive component, positive definite auxiliary interactions recover
its Ising structure, even when the original matrix is singular or
indefinite. For a bipartite component, the interaction between the
two sets of states is a rectangular matrix. The symmetric theory
recovers Ising coordinates on each side separately; equality in a
Fourier norm bound then forces these coordinates to pair across the
interaction, yielding its tensor decomposition.
The two Gram matrices initially determine independent bases on the
two sides. Distance kernels distinguish Fourier degrees, while the
norm equality removes the remaining freedom to mix coordinates
within those degrees.
Finally, finite
moments recover the local factors separated off earlier and determine
which vertex weights preserve tractability, completing the unweighted
and weighted classifications.
Concretely, these moments record the total vertex weight above each
core state at each amplitude. A finite Vandermonde system recovers
those totals, and tractability requires them to agree across the
Ising coordinates.

\paragraph{Concurrent work.}

After we publicly released our initial draft~\cite{liuMeng2026draft},
Cai and Tang informed us of their
independently obtained results~\cite{caiTang2026}.
Cai and Tang prove the same classification for entrywise positive,
positive definite matrices; their v1 states the nonnegative extension
with its proof deferred to Part II. Their route recovers product
structure through an algebra of \emph{centered entrywise logarithms}
and distance-regular graph classification.

\section{Planar simulation tools}\label{sec:tools}

Our graphs may have parallel edges and loops.
Inequalities $M\ge0$ and $M>0$ refer to entries;
$M\succeq0$ and $M\succ0$ mean positive semidefinite and positive
definite, respectively. Write $I_q,J_q$ for the identity and all-ones
matrices and $\mathbf1_q$ for the all-ones vector, omitting the size
when clear.
Loops count twice towards the degree; $c(G)$ is the number of connected
components, including isolated vertices. Write $A\circ B$ for the
entrywise product, $A^{\circ t}$ for an entrywise power, and $A^t$ for
an ordinary matrix power.
For a symmetric matrix $A=\sum_\lambda\lambda P_\lambda$, where
$P_\lambda$ projects orthogonally onto the $\lambda$-eigenspace, we use
$f(A)=\sum_\lambda f(\lambda)P_\lambda$ \cite{higham2008}.
For $A\succ0$, this defines $\log A$ and $A^t$.
Perron--Frobenius theory gives a positive
matrix a positive leading eigenvector, unique up to scale
\cite{hornJohnson2013}.

\paragraph{The logarithm of an Ising tensor.}
Let $M=\gamma T(\boldsymbol\rho)$, with $\gamma>0$ and $0<\rho_r<1$,
and let $X_r$ be the permutation matrix flipping bit $r$.
For some scalar $\alpha$,
\begin{equation}\label{eq:ising-log}
 \log M=\alpha I+\sum_r\frac12\log\frac{1+\rho_r}{1-\rho_r}\,X_r.
\end{equation}
This follows by diagonalizing the two-state factors. The
off-diagonal support of $\log M$ is exactly the Boolean cube.

\paragraph{Planar gadgets.}
During reductions we use \emph{mixed instances}: each edge is labelled
by a matrix from a fixed finite set $\mathcal F$, and a vertex may
carry unary functions $u\in\mathcal F$, each multiplying a coloring's
weight by $u(\sigma(v))$. The fixed background factor $w_{\sigma(v)}$
is counted once per vertex. The resulting evaluation problem is
$\PlGH{\mathcal F,w}$; omitting $w$ means $w=\mathbf1$.

A two-terminal planar gadget has ordered terminals on its outer face.
Its signature entry $(i,j)$ is the sum over internal colors when the
terminal colors are fixed to $i,j$. If their color sets are $X,Y$,
respectively, the signature is an $|X|\times|Y|$ matrix, as when the
terminals lie on opposite sides of a bipartite interaction.
Swapping the terminals transposes the signature.

Background vertex weights are omitted at terminals, so substitution
into an input edge counts them once; explicit terminal unaries are
retained. With
$D=\operatorname{diag}(w)$, series composition gives $ADB$ and
parallel composition gives $A\circ B$. In particular, with unit
background, an $h$-edge path gives $A^h$, while $h$ parallel edges
give $A^{\circ h}$. These are the stretching and thickening gadgets
of \cite{cmy2026}; a loop gives the unary $i\mapsto A_{ii}$.
Substitutions preserve the cyclic order at input vertices.

\paragraph{Joint simulation.}
For fixed background weights $w$ and endpoint color sets, a finite
set $\mathcal B$ is \emph{available} from a fixed language $\mathcal F$ if
\[
 \PlGH{\mathcal F\cup\mathcal B,w}\le_T\PlGH{\mathcal F,w}.
\]
A gadget realizes a signature by substitution;
interpolation simulates it through finitely many oracle queries.
Until Section~\ref{sec:real},
all constants in counting problems are algebraic and computations take
place in fixed number fields. Uniformity for a rational parameter
means polynomial cost in its binary encoding length, denoted by
$\operatorname{bit}(x)$, as well as the instance size.

\subsection{Interpolation and functional identities}

\begin{lemma}[Product and spectral interpolation]\label{lem:interpolation}
Let $B,C$ be fixed algebraic matrices with the same row and column
index sets, and suppose $C$ is available. Assume
\[
 C_{ij}=0\ \Longrightarrow\ B_{ij}=0.
\]
For every $m\ge1$ and two sequences of positions
$(i_r,j_r)_{r=1}^m$, $(k_r,\ell_r)_{r=1}^m$, with repetitions allowed
and $C_{i_rj_r}C_{k_r\ell_r}\ne0$ for every $r$, assume also that
\[
 \prod_{r=1}^m C_{i_rj_r}=\prod_{r=1}^m C_{k_r\ell_r}
 \quad\Longrightarrow\quad
 \prod_{r=1}^m B_{i_rj_r}=\prod_{r=1}^m B_{k_r\ell_r}.
\]
Then $B$ is available. The same statement holds for unary functions.

With unit background weights, an available symmetric matrix $C$ also gives
\[
 \begin{aligned}
 C\succeq0&\quad\Longrightarrow\quad
     \Pi_C:=\sum_{\lambda>0}P_\lambda\ \text{is available},\\
 C\succ0&\quad\Longrightarrow\quad
     C^r\ \text{is available for every fixed }r\in\mathbb Q.
 \end{aligned}
\]
Here $P_\lambda$ is the spectral projector of $C$, so $\Pi_C$ is the
orthogonal projector onto $\operatorname{im}C$.
\end{lemma}
\begin{proof}
For product interpolation, fix a mixed instance $\Omega$ with $m\ge1$ marked occurrences of $B$.
Let $s$ be the number of distinct nonzero entries of $C$; if $s=0$,
both matrices are zero and the claim is immediate.
Replacing each marked occurrence by $h$ parallel copies of $C$ gives
\[
 y_h=\sum_\lambda a_\lambda\lambda^h,\qquad
 \#\{\lambda\}\le\binom{m+s-1}{s-1}.
\]
Here $\lambda$ ranges over length-$m$ products of nonzero $C$-entries,
and $a_\lambda$ sums the remaining weights of the corresponding
colorings. The hypotheses assign a well-defined target product
\[
 b_\lambda=\prod_{r=1}^m B_{i_rj_r}
 \quad\text{whenever}\quad
 \lambda=\prod_{r=1}^m C_{i_rj_r}\ne0.
\]
Colorings using a zero entry of $C$ also contribute zero for $B$.
Vandermonde interpolation at consecutive positive integers $h$
recovers the $a_\lambda$, and the desired value is
\[
 Z(\Omega)=\sum_\lambda a_\lambda b_\lambda.
\]
For unary functions, replace parallel edges by repeated unary factors.

For the spectral assertion, use unit background weights and
replace marked edges by paths of
length $h$ and expand $C^h=\sum_{\lambda>0}\lambda^hP_\lambda$.
The same Vandermonde calculation now groups products of positive
eigenvalues. For these groups the target products are
\[
 b_\lambda=1\quad\text{for }\Pi_C,
 \qquad b_\lambda=\lambda^r\quad\text{for }C^r.
\]
Zero eigenspaces contribute zero to the projector; the power case
has $C\succ0$. The fixed matrix size bounds all product counts,
query sizes, and exact bit lengths polynomially.
\end{proof}

In particular, from $C\ge0$ we obtain its zero--one support matrix
and the zero--one indicators of its smallest positive and largest
entries. Positive unaries admit
fixed rational powers. Known nonzero scalar multiples are harmless,
contributing one scalar for each marked occurrence.

The next lemma turns identities of matrix-valued functions into
finite reductions. This lets us extract a distance or a tensor factor
from a family of gadgets without evaluating the functions near a limit.

\begin{lemma}[Transfer of product identities]\label{lem:transfer}
Let $F(t)$ be a polynomial matrix with fixed algebraic coefficients,
or $F(t)=A^t$ for a fixed algebraic $A\succ0$. Let $B$ be a fixed
algebraic matrix with the same row and column index sets as $F(t)$.
Suppose there is an integer $n_0\ge1$ such that
\[
 F(n)>0\quad(n\ge n_0),
\]
and these integer samples are available by one algorithm of cost
$\operatorname{poly}(N+n)$, where $N$ is the instance size.
For every $m\ge1$ and every two sequences of entry positions,
assume
\[
 \prod_{r=1}^m F_{i_rj_r}(t)\equiv
 \prod_{r=1}^m F_{k_r\ell_r}(t)
 \quad\Longrightarrow\quad
 \prod_{r=1}^m B_{i_rj_r}=\prod_{r=1}^m B_{k_r\ell_r}.
\]
Here $\equiv$ means equality as functions of $t$. Then $B$ is available.

For a target $B(x)$ with rational parameter $x$, require the same
implication for each $x$. Let $K$ contain the fixed source data.
Suppose $B(x)$ is encoded in effective number fields $K_x\supseteq K$
with polynomial-time arithmetic and, for a fixed $c$,
\[
 [K_x:K]\le c,\qquad
 \operatorname{bit}(B(x))+\operatorname{time}(B(x))
       \le\operatorname{poly}(\operatorname{bit}(x)).
\]
Here $\operatorname{bit}$ and $\operatorname{time}$ count the encoding
size and construction time, including the field presentation.
Then $B(x)$ is available uniformly in $\operatorname{bit}(x)$.
\end{lemma}
\begin{proof}
Fix an input $\Omega$ with $m\ge1$ marked occurrences of $B$;
for $m=0$ there is nothing to replace. We choose an integer sample
$F(n_m)$ that preserves the product relations needed for this input.

\emph{1. Bound coincidences between distinct product functions.}
Let $E$ index the matrix entries, write $F_e=F_{ij}$ and $B_e=B_{ij}$
for $e=(i,j)$, and put
\[
 \begin{gathered}
 \mathcal I_m=\left\{\alpha\in\mathbb Z_{\ge0}^{E}:
                         \sum_{e\in E}\alpha_e=m\right\},
 \qquad P_m=|\mathcal I_m|=\binom{m+|E|-1}{|E|-1},\\
 p_\alpha(t)=\prod_{e:\alpha_e>0}F_e(t)^{\alpha_e},
 \qquad b_\alpha=\prod_{e:\alpha_e>0}B_e^{\alpha_e}.
 \end{gathered}
\]
Every nonzero difference $p_\alpha-p_\beta$ has at most $D_m$ real
zeros, where
\[
 D_m=\begin{cases}
 m\delta,&F\text{ is polynomial, }\delta=\max_{i,j}\deg F_{ij},\\[2pt]
 \binom{m+s-1}{s-1}-1,&F(t)=A^t,\ s=\#\{\text{distinct eigenvalues of }A\}.
 \end{cases}
\]
The first bound is the polynomial degree bound. For the second,
expand in the positive eigenvalues of $A$: there are at most
$\binom{m+s-1}{s-1}$ distinct exponential bases. A nonzero sum of
$k$ real exponentials with distinct exponents has at most $k-1$
zeros, by induction using Rolle's theorem.

\emph{2. Choose a sample and apply product interpolation.}
Among all pairs $\alpha,\beta\in\mathcal I_m$, there are at most
$\binom{P_m}{2}D_m$ exceptional real points where nonidentical
product functions agree. Hence the integer interval
\[
 I_m=\left\{n_0,\ldots,n_0+\binom{P_m}{2}D_m\right\}
\]
contains a nonexceptional sample. Evaluate the products and choose
\[
 n_m\in\operatorname*{arg\,max}_{n\in I_m}
                \#\{p_\alpha(n):\alpha\in\mathcal I_m\}.
\]
Identical functions agree at every sample, while an extra coincidence
reduces this count. Thus, for every $\alpha,\beta\in\mathcal I_m$,
\[
 p_\alpha(n_m)=p_\beta(n_m)
 \ \Longleftrightarrow\ p_\alpha\equiv p_\beta
 \ \Longrightarrow\ b_\alpha=b_\beta.
\]
The last implication is the hypothesis of this lemma. Since
$F(n_m)>0$, the length-$m$ Vandermonde construction in the proof of
Lemma~\ref{lem:interpolation}, with $C=F(n_m)$, evaluates $\Omega$.
The bounds $P_m,D_m,n_m=\operatorname{poly}(m)$ ensure polynomially
many exact comparisons and queries of polynomial size. Sample
entries and products have polynomial bit length, so composing with
the assumed simulation of $F(n_m)$ proves availability of $B$.

\emph{3. Uniformity for $B(x)$.}
The sample $n_m$, interpolation bases, and recovered coefficients
depend only on $F$ and $\Omega$. Only the target products change to
\[
 b_\alpha(x)=\prod_{e:\alpha_e>0}B_e(x)^{\alpha_e}.
\]
The stated field, construction-time, and encoding bounds therefore
give total cost $\operatorname{poly}(|\Omega|+\operatorname{bit}(x))$.
\end{proof}

Parameter-dependent spectral calculations may use different bounded-degree
extensions $K_x$ of a fixed number field $K$. Whenever an inner answer
$z_x$ belongs to $K$, the field trace (the sum of its algebraic
conjugates) returns it to that encoding before an outer interpolation:
$z_x=[K_x:K]^{-1}\operatorname{Tr}_{K_x/K}(z_x)$.
Degree and coefficient-size bounds keep this operation
polynomial-time in each sample field.

\subsection{Roots, quotients, and background weights}

\begin{lemma}[One-root restriction]\label{lem:root}
Fix an algebraic symmetric matrix $A$ on $[q]$, algebraic weights
$w\in\mathbb R_{>0}^q$, and a subset $X\subseteq[q]$.
There are fixed rooted planar graphs $(F_1,r_1),\ldots,(F_s,r_s)$
and coefficients $c_1,\ldots,c_s\in\mathbb Q(A,w)$ such that every
connected rooted planar graph $(G,r)$ satisfies
\[
 \sum_{\substack{\sigma:V(G)\to[q]\\\sigma(r)\in X}}
   \prod_{v\in V(G)}w_{\sigma(v)}
   \prod_{uv\in E(G)}A_{\sigma(u),\sigma(v)}
 =\sum_{j=1}^s c_j Z_{A,w}(G\vee F_j).
\]
Here $G\vee F_j$ identifies $r$ with $r_j$. In particular, for every
support-component vertex set $S$ of $A$,
\[
 \PlGH{A[S],w|_S}\le_T\PlGH{A,w}.
\]
\end{lemma}
\begin{proof}
Let $f_F(i)$ be the conditional value of a rooted planar graph $F$
with root color $i$ and root weight omitted. Let $U$ span these vectors,
with inner product $\langle f,g\rangle_w=\sum_iw_if_ig_i$.
Project the indicator $\mathbf1_X$ onto $U$, writing the projection as
$\sum_jc_jf_{F_j}$ in a fixed rooted-signature basis. For the rooted
input $G$,
\[
 \sum_{i\in X}w_if_G(i)
   =\sum_jc_j\langle f_{F_j},f_G\rangle_w
   =\sum_jc_j Z_{A,w}(G\vee F_j).
\]
Root identification preserves planarity.
The basis and coefficients are
fixed; the latter are obtained by a positive definite Gram system.
Restricting a root to a target component restricts every nonzero
coloring of a connected input to that component.
\end{proof}
This is a reduction for one restricted root. It does not make
$\mathbf1_X$ available as a unary that can be imposed independently
at arbitrarily many vertices: the projection need only agree with
$\mathbf1_X$ when paired with rooted signatures.

\begin{lemma}[Tensor Gram matrices and weight removal]\label{lem:weights-remove}
Let $v_1,\ldots,v_q\in\mathbb R^n$ be nonzero and pairwise
nonproportional. For $p=\max(1,q-1)$, the tensors
$v_1^{\otimes p},\ldots,v_q^{\otimes p}$ are linearly independent.

Let $B=B^{\mathsf T}\in\mathbb R^{q\times q}$ and
$w\in\mathbb R_{>0}^q$ have fixed algebraic entries. Suppose the rows
of $B$ are nonzero and pairwise nonproportional, and put
\[
 D=\operatorname{diag}(w),\qquad
 K=(BDB)^{\circ p}\succ0,\qquad u_r(i)=w_i^r.
\]
For every fixed finite language $\mathcal F$ containing $B$ and
every fixed $r\in\mathbb Q$,
\[
 \begin{aligned}
 \PlGH{\mathcal F\cup\{D^{-1},u_r\},w}
     &\le_T\PlGH{\mathcal F,w},\\
 \PlGH{\mathcal F,\mathbf1}&\le_T\PlGH{\mathcal F,w}.
 \end{aligned}
\]
Thus weight removal retains all constraints in $\mathcal F$.
\end{lemma}
\begin{proof}
For $q=1$ independence is immediate. Otherwise choose linear forms
$\ell_{ij}$ vanishing at $v_j$ but not $v_i$.
Their products $\prod_{j\ne i}\ell_{ij}$ separate the tensor powers.

For weight removal, the matrix $BDB$ is the Gram matrix of the weighted rows of $B$;
entrywise power $p$ gives the Gram matrix of their tensor powers.
By the independence just proved, the gadget $K=(BDB)^{\circ p}$
is positive definite.
A chain of $h$ copies has signature
\[
 K(DK)^{h-1}=D^{-1/2}(D^{1/2}KD^{1/2})^hD^{-1/2}.
\]
Spectral interpolation extrapolates from positive integer $h$ to
$h=0$, supplying $D^{-1}$.
A loop contributes $w_i(D^{-1})_{ii}=1$, cancelling the background
without changing any other constraint in $\mathcal F$.
Finally, product interpolation supplies $D^r$; its loops realize
the unary $u_r(i)=w_i^r$.
\end{proof}

If $K_1,\ldots,K_s$ are the identical-row classes of $A$, retain one
row and column per class to obtain $R$, and put
$\bar w_a=\sum_{i\in K_a}w_i$. Direct summation gives
\begin{equation}\label{eq:twins}
 Z_{A,w}(G)=Z_{R,\bar w}(G).
\end{equation}
Zero rows matter only on isolated input vertices. For a zero--one
matrix, the remaining distinct rows are automatically nonproportional,
so Lemma~\ref{lem:weights-remove} transfers unweighted quotient
hardness to every positive weighting. We call identical rows
\emph{twins}; merging them uses their numerical entries.

\subsection{Hardness endpoints and the FKT algorithm}
We use two known planar hardness results: a strictly positive,
invertible symmetric $2\times2$ matrix with unequal diagonal entries
is hard \cite{guoWilliams2020,cmy2026}; and $I_s+J_s$ is hard for
$s\ge3$, by the planar Tutte classification \cite{vertigan2005}.
The latter follows from
$Z_{I_s+J_s}(G)=s^{c(G)}T_G(s+1,2)$, where $T_G$ is the Tutte
polynomial.
Equal-diagonal Boolean matrices are evaluated by the FKT algorithm,
using the high-temperature expansion
\[
 Z_{W(\rho)}(G)=2^{|V(G)|}
 \left(\frac{1+\rho}{2}\right)^{|E(G)|}
 \sum_{\substack{F\subseteq E(G)\\\deg_F(v)\text{ even}}}
       \left(\frac{1-\rho}{1+\rho}\right)^{|F|}.
\]
Replacing each vertex by a planar matching gadget turns the even-degree
constraint into a weighted perfect-matching sum, evaluated by the
FKT algorithm.
This remains valid for $\rho>1$; loops can equivalently be removed
first, since their original diagonal contribution is $1$.

\section{Logarithmic support, entropy, and geometry}\label{sec:geometry}

Fix a source language and a color set $X$, which may be a prescribed
endpoint domain. Let $\mathcal A$ be the
family generated by symmetric planar gadget compositions and the
interpolations of Section~\ref{sec:tools}, with unit background
weights. Each finite subfamily is available: compose the
finitely many reductions in its construction. Assume $\mathcal A$
contains a connected nonnegative positive definite matrix.

For $H\succ0$ let $R_H$ be the graph on $X$ with an edge $ij$
when $i\ne j$ and $(\log H)_{ij}\ne0$. Choose a connected $M\ge0$, $M\succ0$, in
$\mathcal A$ \emph{maximizing $|E(R_M)|$}, and write $L=\log M$, $R=R_M$.
The maximum exists because there are finitely many possible supports
on $X$. The graph $R$ is connected,
since otherwise $L$, and hence $M=e^L$, would be block diagonal.

\subsection{From maximum support to distance kernels}

We may arrange that every edge entry of $L$ is positive. Indeed, for
an even integer $s$ and small positive rational $t$, the matrix
$K_s(t)=(M^t)^{\circ s}$ is nonnegative and, by the Schur product
theorem, positive definite. Moreover,
\[
 (\log K_s(t))_{ij}=t^sL_{ij}^s+O(t^{s+1})\qquad(ij\in E(R)).
\]
All old edges persist positively, and maximality forbids new ones.
Replace $M$ by one such matrix.

\begin{lemma}[Distance kernels]\label{lem:distance}
Let $H\succ0$ be a fixed algebraic matrix on $X$, with $R_H$
connected and
\[
 (\log H)_{ij}>0\qquad(ij\in E(R_H)).
\]
For every $x\in\mathbb Q_{>0}$, the matrix
\[
 E_{R_H}(x)=\bigl(x^{d_{R_H}(i,j)}\bigr)_{i,j\in X}
\]
is available from $H$, uniformly in $\operatorname{bit}(x)$.

For the maximizing support $R=R_M$ chosen above, there is an
$\varepsilon>0$ such that, for every real $0<x<\varepsilon$,
$E_R(x)\succ0$ and
\begin{equation}\label{eq:log-sparse}
 (\log E_R(x))_{ij}=0\qquad(i\ne j,\ ij\notin E(R)).
\end{equation}
\end{lemma}
\begin{proof}
Put $K=\log H$. We apply Lemma~\ref{lem:transfer} to the source
family $F(t)=H^t$ and target $B=E_{R_H}(x)$.
For $h=d_{R_H}(i,j)$,
\[
 (H^t)_{ij}=\frac{(K^h)_{ij}}{h!}t^h+O(t^{h+1}),
 \qquad (K^h)_{ij}>0.
\]
The leading coefficient sums positive products over shortest paths;
diagonal entries have order zero. Orders add in products, so every
product identity for entries of $H^t$ preserves the sum of their
distances and hence the corresponding product of entries of $E_{R_H}(x)$.
Also $H^t>0$ for $t>0$: shift $K$ by a scalar to make it nonnegative
and use the exponential series. Lemma~\ref{lem:transfer} applies,
with integer samples realized by paths. The target powers of $x$
have bounded exponents, giving uniformity.

For the second assertion, take $H=M$, so $R_H=R$.
Since $E_R(0)=I$, it is positive definite near zero, and on an edge
$(\log E_R(x))_{ij}=x+O(x^2)$. For small rational $x$ it therefore
retains every edge of $R$. Maximality gives \eqref{eq:log-sparse};
analyticity extends the identity to real $x$ near zero.
\end{proof}

\begin{figure}[t]
\centering
\begin{tikzpicture}[every node/.style={font=\small},
 dot/.style={circle,draw,fill=white,inner sep=2pt}]
 \node[dot,fill=blue!30,label=left:$i$] (i) at (0,0) {};
 \node[dot,fill=blue!30,label=right:$j$] (j) at (3,0) {};
 \node[dot,label=above:$u$] (u) at (1.5,.85) {};
 \node[dot,label=below:$v$] (v) at (1.5,-.85) {};
 \draw (i)--(u)--(j)--(v)--(i); \draw[blue!60!black,thick] (u)--(v);
 \node[dot,label=below left:$0$] (o) at (6,-1) {};
 \node[dot,label=below right:$a$] (a) at (8,-1) {};
 \node[dot,label=above left:$b$] (b) at (6,1) {};
 \node[dot,label=above right:$a+b$] (c) at (8,1) {};
 \draw[blue!60!black,thick] (o)--(a) (b)--(c);
 \draw (o)--(b) (a)--(c);
 \draw (6.2,-1)--(6.2,-.8)--(6,-.8);
\end{tikzpicture}
\caption{Left: the Wheatstone gadget, with fixed terminal colors $i,j$
and summed internal colors $u,v$. Its central edge contributes the
factor $E_R(x)_{uv}$, which detects adjacency of the two common
neighbors in $R$. Right: an induced square of $R$ after the embedding
\eqref{eq:embedding}, with one vertex translated to the origin;
the opposite vertex is forced to be $a+b$.}
\label{fig:local}
\end{figure}
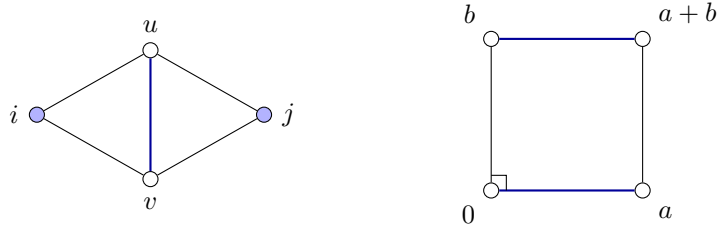

\begin{lemma}[Local rigidity]\label{lem:local}
Let $N\in\mathcal A$ satisfy
\[
 N\succ0,\qquad R_N=R,\qquad K=\log N,\qquad
 K_{uv}>0\quad(uv\in E(R)).
\]
For distinct $i,j\in X$, put $h=d_R(i,j)$ and let
$\mathcal P_{ij}$ be the set of shortest $i$--$j$ paths. Then
\[
 |\mathcal P_{ij}|=h!,\qquad
 \prod_{uv\in P}K_{uv}=\prod_{uv\in Q}K_{uv}
 \quad(P,Q\in\mathcal P_{ij}).
\]
Every distance-two pair has exactly two common neighbors, which
are nonadjacent. Every neighborhood is a disjoint union of cliques.
\end{lemma}
\begin{proof}
\emph{Shortest-path counts and weights.}
For shortest $i$--$j$ paths $P$, put $a_P=\prod_{uv\in P}K_{uv}>0$.
For even $s\ge2$, put $H=(N^t)^{\circ s}$. Near zero its logarithm
retains every edge of $R$, with entry $t^sK_{uv}^s+O(t^{s+1})$.
For rational $t>0$, maximality excludes all other off-diagonal
entries; analyticity gives the same zeros near $t=0$.
The diagonal entries of $\log H$ are $O(t)$. In the expansion
$H=\exp(\log H)$, a term reaching distance $h$ therefore needs
at least $h$ edge factors. The coefficient of $t^{sh}$ uses exactly
$h$ such factors: a diagonal insertion or a longer walk has higher
order. Comparing this shortest-path coefficient with that of
$(N^t)_{ij}^s$ gives
\begin{equation}\label{eq:path-moment}
 \left(\frac{\sum_Pa_P}{h!}\right)^s=\frac{\sum_Pa_P^s}{h!}.
\end{equation}
Taking $s$th roots as $s\to\infty$ shows that the left-hand base is
$a_*:=\max_Pa_P$. Dividing by $a_*^s$ then shows that exactly $h!$
paths attain the maximum. Their sum already equals $h!a_*$, so
positivity excludes every other path.

\emph{Common neighbors and neighborhoods.}
The path count with $h=2$ gives exactly two common neighbors.
To prove they are nonadjacent, use the
Wheatstone gadget in Figure~\ref{fig:local}, with all five edge
matrices equal to $E=E_R(x)$:
\[
 \Omega(E)_{ij}=\sum_{u,v}E_{iu}E_{uj}E_{iv}E_{vj}E_{uv}.
\]
Let $p,q$ be the two common neighbors of the distance-two pair $i,j$,
and set $\varepsilon=\mathbf1_{pq\in E(R)}$. Put
$S=\{i,j,p,q\}$ and $a(u)=d_R(i,u)+d_R(u,j)$. The exponent contributed
by $(u,v)$ is $a(u)+a(v)+d_R(u,v)$, where $a(u)=2$ exactly on $S$.
The terms of degree at most five are
\begin{center}
\begin{tabular}{@{}lcc@{}}
\toprule
Internal colors $(u,v)$ & Number & Degree\\
\midrule
$u=v\in S$ & $4$ & $4$\\
One in $\{i,j\}$, one in $\{p,q\}$ & $8$ & $5$\\
$(p,q),(q,p)$ with $pq\in E(R)$ & $2\varepsilon$ & $5$\\
\bottomrule
\end{tabular}
\end{center}
All other pairs have degree at least six: if one color lies outside
$S$, either $u=v$ or $a(u)+a(v)\ge5$ and $d_R(u,v)\ge1$.
For an edge $ab$, $\Omega(E)_{ab}=2x^2+2x^3+O(x^4)$; diagonal
entries are $1+O(x^3)$, and in general
$\Omega(E)_{ab}=O(x^{2d_R(a,b)})$. Thus only the intermediate colors $p,q$
contribute through degree five to $((\Omega(E)-I)^2)_{ij}$, giving
\begin{align*}
 \Omega(E)_{ij}&=4x^4+(8+2\varepsilon)x^5+O(x^6),\\
 ((\Omega(E)-I)^2)_{ij}&=8x^4+16x^5+O(x^6),\\
 (\log\Omega(E))_{ij}&=2\varepsilon x^5+O(x^6).
\end{align*}
Since $\Omega(E)-I=O(x^2)$, terms of order at least three in the
logarithm series contribute $O(x^6)$.
On every edge $ab$ the logarithmic entry is $2x^2+O(x^3)$, while
$\Omega(E)\to I$. Thus $\varepsilon=1$ would add a logarithmic edge
to a nonnegative positive definite available matrix, contradicting
maximality. An induced three-vertex path in a neighborhood would
give a distance-two pair with adjacent common neighbors. Hence every
neighborhood is a disjoint union of cliques.
\end{proof}

\subsection{Entropy supplies global positive definiteness}\label{sec:entropy}

Recovering Euclidean distance requires control as $x\uparrow1$.
Maximizing matrix entropy $-\operatorname{Tr}(H\log H)$ supplies
a positive definite completion throughout the interval and
identifies it with the distance kernel.

\begin{lemma}[Entropy continuation]\label{lem:entropy}
Let $R$ be a finite connected graph. Suppose there is an
$\varepsilon>0$ such that \eqref{eq:log-sparse} holds for every
$0<x<\varepsilon$. Then
\[
 E_R(x)\succ0\qquad(0<x<1).
\]
\end{lemma}
\begin{proof}
We construct an analytic matrix $H(x)\succ0$ on $(0,1)$ satisfying
the nonedge equations, and then show $H(x)=E_R(x)$.

For each $x\in(0,1)$ the trace is fixed, so we minimize
\begin{equation}\label{eq:entropy}
 \Phi(H)=\operatorname{Tr}(H\log H-H)
 \quad\text{over}\quad
 \mathcal F_x=\{H\succeq0:H_{ii}=1,\ H_{ij}=x\ (ij\in E(R))\},
\end{equation}
using $0\log0=0$. The feasible set is compact and contains the
positive definite matrix $(1-x)I+xJ$. A minimizer is positive
definite: moving a singular point towards this matrix changes the
zero-eigenvalue contribution by $c\epsilon\log\epsilon+O(\epsilon)$,
$c>0$, and all other contributions by $O(\epsilon)$, decreasing
$\Phi$.

On real symmetric matrices use $\langle U,V\rangle=\operatorname{Tr}(UV)$
and the Frobenius norm $\|V\|_F^2=\operatorname{Tr}(V^2)$.
Write $Df_H[V]$ for the derivative of $f$ at $H$ in direction $V$.
The gradient of $\Phi$ is $\log H$, and its Hessian is positive:
\[
 \langle V,D\log_H[V]\rangle
  =\int_0^\infty
       \|(H+sI)^{-1/2}V(H+sI)^{-1/2}\|_F^2\,ds>0
       \quad(V\ne0).
\]
There is therefore a unique minimizer $H(x)$, characterized by
$(\log H(x))_{ij}=0$ on the unconstrained off-diagonal positions.

We next show that $H(x)$ is analytic and agrees with $E_R(x)$.
Let $\mathcal S_R$ be the symmetric matrices supported on the diagonal
and edges of $R$, and $\pi_R$ the orthogonal projection onto this
space, obtained by zeroing the other entries. With $A_R$ the adjacency
matrix, the stationarity equations are
\[
 C(x)=\log H(x)\in\mathcal S_R,\qquad
 \pi_R(e^{C(x)})=I+xA_R.
\]
For nonzero $V\in\mathcal S_R$,
\[
 \langle V,D(\pi_R\exp)_C[V]\rangle
       =\int_0^1\operatorname{Tr}(Ve^{(1-s)C}Ve^{sC})\,ds>0.
\]
The derivative is invertible, so the implicit-function theorem and
uniqueness make $H(x)$ analytic throughout $(0,1)$. Near zero,
$E_R(x)$ is positive definite and satisfies the same stationarity
equations, hence equals $H(x)$. The analytic identity theorem gives
\[
 E_R(x)=H(x)\succ0\qquad(0<x<1),
\]
which is the required conclusion.
\end{proof}

\begin{remark*}[Quantum maximum entropy]
Under Lemma~\ref{lem:entropy}, put $q=|V(R)|$ and
$\rho_x=E_R(x)/q$, $0<x<1$. A density matrix is a complex
Hermitian positive semidefinite matrix of trace one. Among such
$q\times q$ matrices $\sigma$ satisfying
\[
 \sigma_{ii}=1/q,\qquad
 \operatorname{Re}\sigma_{ij}=x/q\quad(ij\in E(R)),
\]
$\rho_x$ uniquely maximizes the von Neumann entropy
$S(\sigma)=-\operatorname{Tr}(\sigma\log\sigma)$, with natural
logarithms and $0\log0=0$. Indeed, the stationarity equations above
give $\operatorname{Tr}((\sigma-\rho_x)\log\rho_x)=0$. Hence
\begin{align*}
 D(\sigma\Vert\rho_x)
   &:=\operatorname{Tr}\bigl(\sigma(\log\sigma-\log\rho_x)\bigr)
     =S(\rho_x)-S(\sigma),\\
 \|\sigma-\rho_x\|_1&\le\sqrt{2D(\sigma\Vert\rho_x)}.
\end{align*}
Here $D$ is quantum relative entropy and $\|A\|_1$ is the sum of
the singular values of $A$; the inequality is quantum Pinsker
\cite[Theorem~5.38]{watrous2018}. Thus an entropy deficit
$\varepsilon$ bounds the trace-norm error by $\sqrt{2\varepsilon}$
under the same constraints. For a Boolean $d$-cube, the recovered color
coordinates give the product state $\rho_x=(W(x)/2)^{\otimes d}$.
\end{remark*}

\subsection{Euclidean distance forces a Cartesian product}

Write $K_s$ for the complete graph on $s$ vertices.
In a Cartesian product $R_1\square\cdots\square R_d$, adjacent
tuples differ along one factor edge.

\begin{lemma}[Geometric rigidity]\label{lem:cartesian}
Let $R$ be a finite connected graph satisfying:
\begin{enumerate}[label=\textup{(\roman*)},leftmargin=2.2em]
\item $E_R(x)\succeq0$ for every $x\in(0,1)$;
\item every neighborhood is a disjoint union of cliques;
\item every distance-two pair has exactly two common neighbors,
      which are nonadjacent.
\end{enumerate}
Then
\[
 R\cong K_{s_1}\square\cdots\square K_{s_d},\qquad
 d\ge0,\quad s_r\ge2.
\]
For $d=0$, the product is a single vertex.
\end{lemma}
\begin{proof}
To embed the graph distance, put $X=V(R)$, $\Delta=(d_R(i,j))_{i,j\in X}$, and $P=I-J/|X|$. Since
$E_R(x)=J-(1-x)\Delta+O((1-x)^2)$, positivity gives
$z^{\mathsf T}\Delta z\le0$ for $z\perp\mathbf1$.
Thus $-P\Delta P/2$ is a Gram matrix. Taking vectors with this Gram
matrix and expanding their squared differences gives an injective map
$f$ such that
\begin{equation}\label{eq:embedding}
 \|f(u)-f(v)\|^2=d_R(u,v).
\end{equation}
An induced square $u,v,z,w,u$ becomes a parallelogram:
the vectors $a=f(v)-f(u)$ and $b=f(w)-f(u)$ are orthogonal unit
vectors, and the prescribed distances force
$f(z)=f(u)+a+b$.

We next use this parallelogram rule to transport cliques across
an edge. Let $N(u)$ denote the neighbors of $u$. Every edge $uv$ lies in a
unique maximal clique $A$, because $N(u)$ is a disjoint union of cliques.
For each $w\in N(u)\setminus A$, the unique other common neighbor
$z$ of $v,w$ completes such a square. This gives a reversible
bijection from $N(u)\setminus A$ to $N(v)\setminus A$, translating
$f(w)$ by $f(v)-f(u)$. It transports entire maximal cliques;
an extra vertex in an image clique would enlarge the original one
under the inverse bijection.

To recover the Cartesian coordinates, choose a base vertex $o$ with $f(o)=0$, and let $S_r=f(C_r)$ for
the maximal cliques through $o$. Each $S_r$ is a unit-side regular
simplex containing zero, and their spans are orthogonal by
\eqref{eq:embedding}. We maintain the following invariant along a walk:
\[
 f(v)=\sum_rz_r,\quad z_r\in S_r,\qquad
 \{f(C):C\text{ maximal through }v\}
     =\{f(v)+(S_r-z_r):r\}.
\]
It holds at $o$. A step in the $r$th clique replaces $z_r$ by
another point of $S_r$. That clique stays fixed, while the
parallelogram rule translates every other clique by the step vector,
preserving the invariant. Orthogonal projection onto
$\operatorname{span}(S_r)$ recovers $z_r$ from $f(v)$, so different
walks to $v$ give the same coordinates. Connectedness gives
$f(V(R))\subseteq\sum_rS_r$. Conversely, the invariant lets us
change each coordinate to any point of $S_r$, giving every tuple.
Thus coordinates are unique and adjacency is exactly one coordinate change. Hence
$R=\square_r K_{|S_r|}$.
\end{proof}

\begin{theorem}[Maximum logarithmic support]\label{thm:geometry}
The maximizing graph $R$ has the form
\[
 R\cong K_{s_1}\square\cdots\square K_{s_d},\qquad
 d\ge0,\quad s_r\ge2.
\]
\end{theorem}
\begin{proof}
Combine Lemmas~\ref{lem:distance}, \ref{lem:local},
\ref{lem:entropy}, and~\ref{lem:cartesian}.
\end{proof}

\subsection{Recovering matrices in common coordinates}

Write $Q_d=K_2\square\cdots\square K_2$ for the Boolean cube.
The geometry so far describes one logarithmic support. When it is
$Q_d$, the next proposition places every available matrix in the same
coordinates, as needed to compare Gram matrices later.

\begin{proposition}[Common tensor coordinates]\label{prop:coordinates}
Suppose $R\cong Q_d$ and fix an identification of its vertices
with $\{0,1\}^d$. For every $N\in\mathcal A$ with $N\ge0$ and
$N\succ0$, there is a factorization in these fixed coordinates
\[
 N=\gamma\bigotimes_{r=1}^d C_r,\qquad
 \gamma>0,\quad C_r\in\mathbb R_{\ge0}^{2\times2},\quad C_r\succ0.
\]
If the support of $N$ is connected, then $C_r>0$ for every $r$.
If $N$ is algebraic, $\gamma$ and all $C_r$ can be chosen algebraic.
\end{proposition}
\begin{proof}
We first determine the off-diagonal entries of $\log N$.
Index colors by $z\in\{0,1\}^d$; let $e_r$ be the $r$th unit bit
vector, with bit additions taken modulo two. Put $B=\log N$.
For small rational $t,\eta>0$, the available matrix
\[
 A(t,\eta)=e^{tL/2}e^{\eta tB}e^{tL/2}
\]
is positive definite. Its entrywise square is nonnegative and has
logarithmic leading term $t^2(L+\eta B)_{ij}^2$ off the diagonal.
Choose $\eta$ small enough to retain every edge of $R$. Maximality
then forces $B_{ij}=0$ outside the cube edges.

For sufficiently small positive $t,\eta$, $A(t,\eta)$ is
entrywise positive: since $L,B$ are supported on the cube, the
leading term of its $(i,j)$ entry has order $t^{d_R(i,j)}$, with
positive coefficient for small $\eta$ by continuity from $\eta=0$.
Its logarithm also retains all cube edges positively. For rational
parameters, availability and maximality therefore give
$R_{A(t,\eta)}=R$; continuity gives the same zeros throughout this
parameter region. We can now apply \eqref{eq:path-moment} to both diagonals
of a cube square:
if its cyclic edge weights are $a,b,c,e$, then $ab=ec$ and $ae=bc$,
so $a=c$, $b=e$. Comparing the coefficient of $\eta t$ shows that
$B$ has a constant weight $b_r$ in each direction. With
$\phi(z)=B_{zz}$, this gives
\[
 B=\operatorname{diag}\phi+\sum_r b_rX_r.
\]
We next prove that $\phi$ is a sum of one-bit functions.
Fix rational $x\in(0,1)$ and put
$C=\log E_{Q_d}(x)=cI+b\sum_rX_r$, with $b>0$.
The same maximality argument makes the nonedge entries of
$\log(e^{tC/2}e^{uB}e^{tC/2})$ vanish in an open sector, hence as
analytic identities near zero. Direct expansion to total degree three
gives the following coefficient, where $[C,B]=CB-BC$ and
$[t^2u]$ denotes coefficient extraction:
\[
 [t^2u]\log(e^{tC/2}e^{uB}e^{tC/2})=-\frac1{24}[C,[C,B]].
\]
For opposite square vertices $z'=z+e_r+e_s$,
\[
 [C,[C,B]]_{z,z'}
 =2b^2\bigl(\phi(z)+\phi(z')-\phi(z+e_r)-\phi(z+e_s)\bigr).
\]
The nonedge coefficient vanishes, so this mixed second difference is zero.
Thus each one-bit increment is independent of the other bits:
$\phi(z)=\phi(0)+\sum_r(\phi(e_r)-\phi(0))z_r$.

Finally, the one-bit summands of $B$ commute, so exponentiating gives a
tensor product of positive definite factors. Normalizing their
$(0,0)$ entries to one gives
\[
 \gamma=N_{0,0},\qquad
 C_r=\begin{pmatrix}
 1&N_{0,e_r}/N_{0,0}\\
 N_{0,e_r}/N_{0,0}&N_{e_r,e_r}/N_{0,0}
 \end{pmatrix}.
\]
These entries are nonnegative, and algebraic when $N$ is algebraic.
A zero off-diagonal entry in a factor disconnects the support.
\end{proof}

\section{The positive definite dichotomy}\label{sec:hardness}

The geometry leaves two possible obstructions to tractability:
clique coordinates of size at least three, and Boolean factors with
unequal diagonals. In each case one must recover a hard factor from
a product of partition functions. The following arguments perform
that recovery.

\begin{lemma}[Isolating a Potts factor]\label{lem:potts}
If a maximizing support has a factor $K_s$ with $s\ge3$, the source
problem is $\#\mathrm P$-hard.
\end{lemma}
\begin{proof}
We recover $Z_{I_s+J_s}$ from the source, reducing the hard Potts
problem to it. Write $R=\square_rK_{s_r}$ and
$E_R(t)=\bigotimes_r W_{s_r}(t)$, where
$W_k(t)=(1-t)I_k+tJ_k$. A two-edge path gives $F(t)=E_R(t)^2$.
In a size-$k$ coordinate its diagonal and off-diagonal entries are
\[
 p_k(t)=1+(k-1)t^2,\qquad q_k(t)=2t+(k-2)t^2.
\]
If size $k$ occurs $n_k$ times, a product of $m$ entries of $F(t)$
has the form $\prod_kp_k(t)^{a_k}q_k(t)^{mn_k-a_k}$.
The polynomials $p_k$ are distinct irreducible quadratics over
$\mathbb R$, whereas every $q_k$ splits into real linear factors.
Unique factorization therefore preserves each $a_k$ in a product
identity. Lemma~\ref{lem:transfer} supplies
\[
 H_s=(I_s+J_s)^{\otimes n_s}\otimes J_\ell,
 \qquad \ell=\prod_{r:s_r\ne s}s_r,
\]
by replacing a product by $2^{a_s}$.
The recovered value satisfies
\[
 Z_{I_s+J_s}(G)
   =\left(\frac{Z_{H_s}(G)}{\ell^{|V(G)|}}\right)^{1/n_s}.
\]
The root is the unique positive one. Exact algebraic root recovery
is polynomial-time, and the Potts endpoint is hard.
\end{proof}

\begin{lemma}[Isolating a Boolean field]\label{lem:boolean-field}
Let
\[
 N=\bigotimes_{r=1}^d C_r,\qquad
 C_r=\begin{pmatrix}a_r&b_r\\b_r&c_r\end{pmatrix},\qquad
 a_r,b_r,c_r>0,\quad a_rc_r>b_r^2,
\]
with fixed algebraic parameters. If $a_r\ne c_r$ for some $r$,
then $\PlGH{N}$ is $\#\mathrm P$-hard.
\end{lemma}
\begin{proof}
We recover a Boolean interaction with unequal diagonal entries.

\emph{An available tensor family.}
After a scalar normalization and independent bit flips, write
\[
 C_r=\begin{pmatrix}\theta_r&w_r\\w_r&\theta_r^{-1}\end{pmatrix},
 \qquad \theta_r\ge1,\quad 0<w_r<1.
\]
The logarithm of $N$ has positive weights on cube edges, so
Lemma~\ref{lem:distance} supplies $E_{Q_d}(x)$. For any fixed
positive odd $t$, add $(t-1)/2$ source loops at each terminal of an
$N$-edge, then take a parallel product with this kernel. This supplies
\begin{equation}\label{eq:field-family}
 F_t(x)=\bigotimes_r B_r(x),\qquad
 B_r(x)=\begin{pmatrix}\theta_r^t&w_rx\\w_rx&\theta_r^{-t}\end{pmatrix}.
\end{equation}
This simulation is uniform for rational $0<x<1$.

Group identical pairs $(\theta_r,w_r)$ into classes $g$ of size
$n_g$. Set $a_g=(\theta_g^t-\theta_g^{-t})/2$ and
$c_g=(\theta_g^t+\theta_g^{-t})/2$. Choose $t$ so that
$\beta_g=a_g/w_g$ are distinct among classes with $\theta_g>1$.
Classes with equal $\theta_g$ already have distinct $w_g$.
For different values $\theta_g=e^h,\theta_{g'}=e^k$, the ratio
$\sinh(th)/\sinh(tk)$ is strictly monotone in $t>0$, as follows
from the strict increase of $z\coth z$. Each pair excludes at
most one value of $t$, leaving an admissible positive odd integer.

\emph{Product identities that retain one factor.}
We show that spectral product identities preserve each
unequal-diagonal class separately. The eigenvalue branches and
their ratio are
\[
 \lambda_g^\pm(x)=c_g\pm\sqrt{a_g^2+w_g^2x^2},\qquad
 \delta_g(x)=\lambda_g^-(x)/\lambda_g^+(x).
\]
Both eigenvalues are positive on $(0,1)$. A product of $m$ eigenvalues
of the full tensor $F_t(x)$ has the form
\[
 \prod_g(\lambda_g^+(x))^{mn_g}\delta_g(x)^{k_g},
 \qquad 0\le k_g\le mn_g.
\]
The common $\lambda_g^+$ factors cancel when comparing two such
products, leaving $\prod_g\delta_g^{z_g}\equiv1$ with integer $z_g$.
For $\theta_g>1$, the
square root has a simple branch point at $x=\mathrm i\beta_g$,
where $\mathrm i^2=-1$. Continuing once around that point changes
the sign of this square root and leaves the others unchanged, since
the $\beta_g$ are distinct. Thus it exchanges just the two branches
of class $g$. Applying this continuation to the product identity gives
$\delta_g^{2z_g}\equiv1$.
Since $0<\delta_g<1$, we obtain $z_g=0$. Thus an identity between
equal-length products of eigenvalues of $F_t(x)$ preserves the
number of smaller-branch choices in each unequal-diagonal class.

Fix one such class $g$. These identities permit the spectral
replacement
\[
 T_g(x)=B_g(x)^{\otimes n_g}\otimes I_2^{\otimes(d-n_g)}
\]
in the tensor eigenbasis. This retains the chosen Boolean interaction
and replaces every other factor by equality. To implement it, we
select samples at which the functional product identities are the
only numerical coincidences.

\emph{Implementing the replacement with finite queries.}
Let $K$ be a fixed number field containing all $\theta_h,w_h$.
For an input with $m$ edges, there are at most
$P_m=\binom{m+2^d-1}{2^d-1}$ formal spectral products.
They lie in the fixed-degree extension
\[
 K\bigl(x,\sqrt{a_h^2+w_h^2x^2}:\theta_h>1\bigr)/K(x).
\]
If there are $b$ unequal-diagonal classes, multiply a nonzero
difference by all its algebraic conjugates (its field norm).
The result is a nonzero polynomial of degree at most $D_m=2^bdm$:
the spectral products are integral over $K[x]$, each conjugate
grows at most as $x^{dm}$, and there are at most $2^b$ conjugates.
Every zero of the difference is a zero of this polynomial. Hence there are at most
$B_m=\binom{P_m}{2}D_m$ exceptional real samples at which two
nonidentical products coincide.

\Needspace{13\baselineskip}
The reduction now has three finite steps.
\begin{enumerate}[label=\textup{\arabic*.},leftmargin=2em,itemsep=2pt]
\item \emph{Separate products.} Choose $B_m+n_gm+1$ distinct rational
points in $(0,1)$ of polynomial bit length. Retain those with the
maximum number of distinct numerical spectral products. These are
exactly the nonexceptional samples, so at least $n_gm+1$ remain.
\item \emph{Recover the chosen factor at each sample.} For each retained
$x$, replace input edges by paths of consecutive lengths and use
Vandermonde interpolation to recover the coefficients grouped by
spectral product. Assign each group its product for $T_g(x)$,
as permitted by the branch argument. This gives $Z_{T_g(x)}(G)$;
return it to $K$ by the trace operation in Section~\ref{sec:tools}.
\item \emph{Fix the target.} The value $Z_{T_g(x)}(G)$ is a polynomial
in $x$ of degree at most $n_gm$. Interpolate these values to evaluate
it at $x_0=1/2$.
\end{enumerate}
The sample fields have bounded degree, and all sample counts, path
lengths and bit lengths are polynomial. Each $F_t(x)$ is simulated
from $N$, so every oracle call evaluates the original fixed
interaction on a planar graph.

Finally,
\[
 Z_{T_g(x_0)}(G)
   =2^{(d-n_g)c(G)}Z_{B_g(x_0)}(G)^{n_g}.
\]
Recovering the positive root evaluates the invertible Boolean
matrix $B_g(x_0)$, whose diagonal is unequal. The Boolean hardness
endpoint completes the reduction.
\end{proof}

\begin{theorem}[Positive definite core]\label{thm:pd}
Let $M=M^{\mathsf T}\in\mathbb R_{\ge0}^{q\times q}$ have fixed
algebraic entries, connected support, and $M\succ0$.
Then $\PlGH{M}\in\mathsf{FP}$ if
\[
 q=2^d,\qquad M\cong\gamma T(\boldsymbol\rho),\qquad
 d\ge0,\quad\gamma>0,\quad\boldsymbol\rho\in(0,1)^d,
\]
and $\PlGH{M}$ is $\#\mathrm P$-hard otherwise.
\end{theorem}
\begin{proof}
Assume $\PlGH{M}$ is not $\#\mathrm P$-hard, and choose a maximum
logarithmic support in the available family.
Theorem~\ref{thm:geometry} and Lemma~\ref{lem:potts} leave a cube.
Proposition~\ref{prop:coordinates} factors $M$ in its coordinates,
and Lemma~\ref{lem:boolean-field} forces each factor to have equal
diagonals. Positive definiteness then gives $0<\rho_r<1$.
Conversely, the partition function factors into zero-field Ising
partition functions, each evaluated by the FKT algorithm.
\end{proof}

\begin{corollary}[A whole family in the same coordinates]\label{cor:common}
Suppose the source is not $\#\mathrm P$-hard and its available
family $\mathcal A$ contains a matrix $M$ with connected support and
\[
 M\ge0,\qquad M\succ0.
\]
Then the color set has size $2^d$ and admits one identification
with $\{0,1\}^d$ such that every $N\in\mathcal A$ with $N>0$ has,
in these same coordinates,
\[
 N=cT(\boldsymbol\rho),\qquad c>0,\quad
 \boldsymbol\rho\in\mathbb R_{>0}^d.
\]
\end{corollary}
\begin{proof}
The preceding proof and Proposition~\ref{prop:coordinates} first
give common coordinates for the positive definite members.
For general $N>0$, the available matrix
$N\circ E_{Q_d}(x)$ tends to the positive diagonal
$\operatorname{diag}(N_{ii})$ as $x\downarrow0$.
For small rational $x>0$ it is therefore $cT(\boldsymbol\eta)$
in these coordinates. Entrywise comparison gives
$N=cT(\eta_1/x,\ldots,\eta_d/x)$.
\end{proof}

\section{Nonnegative interactions and vertex weights}\label{sec:nonnegative}

We now remove the definiteness assumption. Support interpolation
first limits the possible zero patterns. Normalization then separates
amplitudes from a core with nonproportional rows. A symmetric core
is controlled by the common-coordinate result; a rectangular core
requires its two sets of coordinates to be aligned. Finally, moments
restore the amplitudes and identify the admissible vertex weights.

\subsection{The zero--one classification}

\begin{proposition}[Support rigidity]\label{prop:zero-one}
Let $A=A^{\mathsf T}\in\{0,1\}^{q\times q}$ and let
$w\in\mathbb R_{>0}^q$ be fixed and algebraic. A support component
of $A$ is \emph{basic} if it is a fully looped clique, a loopless
complete bipartite graph, or a loopless isolated vertex.
If every support component is basic, then
\[
 \PlGH{A}\in\mathsf{FP},\qquad \PlGH{A,w}\in\mathsf{FP}.
\]
Otherwise $\PlGH{A}$ is $\#\mathrm P$-complete and
$\PlGH{A,w}$ is $\#\mathrm P$-hard.
\end{proposition}
\begin{proof}
We first prove hardness for nonbasic supports.
Restrict to a connected component; a zero singleton is already
basic. On every other component, merge equal rows.
Lemmas~\ref{lem:root} and~\ref{lem:weights-remove}
transfer hardness back to the original weighted problem.
Let $q$ be the order of the resulting zero--one matrix $A$. It has
distinct nonzero rows, so
$Q=(A^2)^{\circ p}\succ0$ for $p=\max(1,q-1)$.
Unless some connected block of $Q$ is hard, Theorem~\ref{thm:pd}
makes every block of $Q$ an Ising tensor. Positive entrywise
$p$th roots act factor by factor, replacing each parameter
$\eta_r\in(0,1)$ by $\eta_r^{1/p}\in(0,1)$. Thus every block of
$A^2$ has the form
\begin{equation}\label{eq:integer-gram}
 c\bigotimes_{r=1}^{\ell}W(\rho_r),\qquad 0<\rho_r<1.
\end{equation}
In particular these blocks are positive definite with constant
diagonal.

The components of $A^2$ are the even-walk classes of the connected
support of $A$: one class in the nonbipartite case, or the two parts
in the bipartite case. In the former case the constant diagonal
$A^2_{ii}=\sum_jA_{ij}$ makes $A$ regular, say of degree $k$.
In the latter, write $A=\left(\begin{smallmatrix}0&V\\V^{\mathsf T}&0\end{smallmatrix}\right)$.
Positive definiteness of both $VV^{\mathsf T}$ and $V^{\mathsf T}V$
makes the two parts equally large; their constant degrees then agree
by counting edges. Again $A\mathbf1=k\mathbf1$.

If every block of $A^2$ is a singleton, distinct rows of $A$ have
disjoint supports, and symmetry forces each row to have at most one
nonzero entry. Connectedness leaves the basic cases. Otherwise take
a nonsingleton integer block $B$ of $A^2$. Its diagonal is $k$ and
row sum is $k^2$, so \eqref{eq:integer-gram} becomes
$B=k\bigotimes_{r=1}^{\ell}W(s_r/t_r)$, with $\ell\ge1$ and
coprime integers $0<s_r<t_r$. All numbers
$k\prod_{r\in S}s_r/t_r$ are integers. For a prime $p$, take
$S_p=\{r:p\mid t_r\}$. Each selected numerator is coprime to $p$,
so integrality for $S_p$ gives
\[
 v_p(k)\ge\sum_{r\in S_p}v_p(t_r)=\sum_rv_p(t_r),
\]
where $v_p$ is the exponent of $p$. Thus
$\tau:=\prod_rt_r\mid k$. But the row sum gives
\[
 k\tau=\prod_r(t_r+s_r)<\prod_rt_r^2=\tau^2,
\]
contradicting $k\ge\tau$. Thus a nonbasic support yields a hard
block.

Conversely, basic components are evaluated by independent color sums and
the two bipartite orientations. The unweighted count belongs to
$\#\mathrm P$, giving completeness.
\end{proof}

Consequently a nonnegative matrix, with positive background weights,
is hard unless every component is a positive symmetric block, a
positive rectangular bipartite block, or a zero singleton.

\subsection{Symmetric and rectangular cores}

\begin{lemma}\label{lem:positive-core}
Let $B=B^{\mathsf T}\in\mathbb R_{>0}^{q\times q}$ have fixed
algebraic entries, distinct rows, and $B_{ii}=1$ for all $i$.
Either $\PlGH{B}$ is $\#\mathrm P$-hard, or
\[
 q=2^d,\qquad B\cong T(\boldsymbol\rho),\qquad
 \rho_r>0,\quad\rho_r\ne1\quad(r\in[d]).
\]
\end{lemma}
\begin{proof}
Two proportional rows $B_i=cB_j$ would give $1=cB_{ij}$ and
$B_{ij}=c$, hence $c=1$. Thus $(B^2)^{\circ p}$ is positive
definite for some fixed $p$. Corollary~\ref{cor:common} applies
to the family generated from $B$, giving $B=cT(\boldsymbol\rho)$
unless the source is hard. Unit diagonal sets $c=1$, and distinct
rows exclude $\rho_r=1$.
\end{proof}

\paragraph{Rectangular cores.}

For $V>0$ write $A_V=\left(\begin{smallmatrix}0&V\\V^{\mathsf T}&0\end{smallmatrix}\right)$,
on disjoint color sets $X,Y$. Counting on a prescribed input
bipartition $L\to X$, $R\to Y$ is Turing equivalent to the ordinary
$A_V$ problem: one root restriction determines the orientation of
a connected bipartite input, and summing the two orientations
recovers the ordinary value. Nonbipartite inputs contribute zero;
isolated vertices contribute the appropriate weight sums. To use a
same-side Gram $VV^{\mathsf T}$ on an arbitrary planar graph,
subdivide every input edge and sum the new vertex over $Y$, keeping
the original vertices in $X$. The subdivision supplies a bipartite
instance with the required prescribed orientation.

\begin{proposition}[Rectangular rigidity]\label{prop:rectangular}
Let $B\in\mathbb R_{>0}^{p\times q}$ have fixed algebraic entries,
pairwise nonproportional rows, and pairwise nonproportional columns.
Either $\PlGH{A_B}$ is $\#\mathrm P$-hard, or there are permutation
matrices $P,Q$ such that
\[
 p=q=2^d,\qquad PBQ^{\mathsf T}=cT(\boldsymbol\rho),\qquad
 c>0,\quad\boldsymbol\rho\in(0,1)^d.
\]
\end{proposition}
\begin{proof}
Assume the source is not hard. Tensor Gram matrices on both sides
give positive definite available matrices. Choose the common cube
coordinates of Corollary~\ref{cor:common} separately on $X,Y$.
For a fixed $h$, $(BB^{\mathsf T})^{\circ h}$ and
$(B^{\mathsf T}B)^{\circ h}$ are positive definite Ising tensors.
For either Gram matrix $G$, taking roots factor by factor gives
\[
 G^{\circ h}=c\bigotimes_rW(\eta_r)
 \quad\Longrightarrow\quad
 G=c^{1/h}\bigotimes_rW(\eta_r^{1/h})\succ0,
 \qquad 0<\eta_r<1.
\]
Hence $B$ is square of order $n=2^d$.
For $n=1$ the conclusion is immediate; assume $d\ge1$.
Their constant Perron vectors imply
$B\mathbf1=B^{\mathsf T}\mathbf1=s_0\mathbf1$ for some $s_0>0$:
the two Grams have the same largest eigenvalue, and $B$ carries
one positive Perron vector to the other.

Let $H_X,H_Y$ be the orthogonal Walsh matrices in the chosen
coordinates, with entries $n^{-1/2}\chi_S(x)$,
$\chi_S(x)=(-1)^{\sum_{i\in S}x_i}$ for $S\subseteq[d]$.
The \emph{Fourier degree} of this basis vector is $|S|$; the empty
set indexes the constant vector. Put
\[
 C=s_0^{-1}H_X^{\mathsf T}BH_Y=\Theta O,\qquad
 \Theta=\operatorname{diag}(\theta_S),\quad
 \theta_S=\prod_{i\in S}\theta_i,\quad 0<\theta_i<1,
\]
where $O$ is orthogonal. Indeed, $CC^{\mathsf T}$ is diagonal,
and its tensor eigenvalues are $\theta_S^2$; dividing each row by
$\theta_S$ leaves an orthogonal matrix. The empty-set row and column
of $C$ have sole entry $1$ at the empty set. The remaining freedom
is $O$; we show that relabeling the $Y$ bits makes it the identity.

\emph{Noise separates Fourier degrees.}
Distance kernels supply the noise matrix
\[
 K(z)=H\operatorname{diag}(z^{|S|})H^{\mathsf T}
     =\left(\frac{1+z}{2}\right)^d
        E_{Q_d}\!\left(\frac{1-z}{1+z}\right),\qquad 0<z<1.
\]
The subscripts $X,Y$ specify the coordinates on the corresponding
side. Thus $K(z)$ scales a degree-$r$ character by $z^r$.
For rational $z$, the available matrix $s_0^{-2}BK_Y(z)B^{\mathsf T}$
is an Ising tensor in the fixed $X$ coordinates. Thus
$O\operatorname{diag}(z^{|S|})O^{\mathsf T}$ is diagonal.
For fixed $0<z<1$, different degrees have distinct eigenvalues;
each row of $O$ is an eigenvector of $\operatorname{diag}(z^{|T|})$ and is
therefore supported on a single degree $h(S)$. The normalized
Gram eigenvalues are
\[
 \lambda_S(z)=\theta_S^2z^{h(S)},\qquad
 \lambda_\varnothing(z)=1.
\]
The tensor identities $\lambda_S=\prod_{i\in S}\lambda_{\{i\}}$
and $\theta_S=\prod_{i\in S}\theta_i$ give
$h(S)=\sum_{i\in S}h(\{i\})$. The empty row occupies the unique
degree-zero direction. Every other row is orthogonal to it, so
$h(\{i\})\ge1$. Since $h([d])\le d$, all
$h(\{i\})=1$ and $h(S)=|S|$. We have
\begin{equation}\label{eq:fourier-layers}
 C_{S,T}=\begin{cases}
 \theta_S(U_{|S|})_{S,T},&|S|=|T|,\\0,&|S|\ne|T|,
 \end{cases}
\end{equation}
with each $U_r$ an orthogonal matrix indexed by the $r$-element
subsets of $[d]$.

\emph{Parallel edges force equality in a norm bound.}
Apply noise on $X$ and normalize:
\[
 F_\epsilon(x,y)=\frac n{s_0}(K_X(\epsilon)B)_{x,y}
  =\sum_{|S|=|T|}\epsilon^{|S|}\theta_S
                   (U_{|S|})_{S,T}\chi_S(x)\chi_T(y).
\]
Define the double Fourier coefficients by
\[
 F_\epsilon(x,y)^2=\sum_{S,T}D_{S,T}(\epsilon)\chi_S(x)\chi_T(y),
 \qquad H_X^{\mathsf T}F_\epsilon^{\circ2}H_Y=nD(\epsilon).
\]
Character multiplication, $\chi_I\chi_J=\chi_{I\triangle J}$,
where $\triangle$ denotes symmetric difference, gives the convolution
\[
 D_{S,T}(\epsilon)
 =\sum_{\substack{I\triangle J=S\\A\triangle B=T}}
       \epsilon^{|I|+|J|}C_{I,A}C_{J,B}.
\]
Since $|I|+|J|=|S|+2|I\cap J|$, the coefficient of
$\epsilon^{|S|}$ comes exactly from ordered disjoint decompositions
$I\cup J=S$.
Two parallel edges on each side of an internal
$Y$ vertex supply the positive symmetric matrix
$L_\epsilon=n^{-2}F_\epsilon^{\circ2}(F_\epsilon^{\circ2})^{\mathsf T}$.
It is an Ising tensor in the same $X$ coordinates, even if singular,
by Corollary~\ref{cor:common}. Its Fourier matrix is $DD^{\mathsf T}$,
so its eigenvalues are the squared row norms
$\mu_S=\sum_TD_{S,T}^2$. Tensor eigenvalues therefore satisfy, for
$S\ne\varnothing$,
\[
 \mu_S\mu_\varnothing^{|S|-1}=\prod_{i\in S}\mu_{\{i\}}.
\]
These identities hold for rational $0<\epsilon<1$, hence identically
as polynomials in $\epsilon$. Since $\mu_\varnothing(0)=1$ and
$\mu_{\{i\}}=4\theta_i^2\epsilon^2+O(\epsilon^3)$, comparing
lowest-order terms gives
\begin{equation}\label{eq:leading-norm}
 \mu_S=4^{|S|}\theta_S^2\epsilon^{2|S|}
           +O(\epsilon^{2|S|+1}).
\end{equation}

For any matrix $A$, write $A_{S,*}$ for its row indexed by $S$,
and use Euclidean row norms. For $U=U_1$, abbreviate singleton
indices by their elements and define
$P_2(U)_{\{i,k\},\{j,l\}}=U_{ij}U_{kl}+U_{il}U_{kj}$.
The leading degree-two row of $D$ is
$2\epsilon^2\theta_i\theta_k((U_2)_{\{i,k\},*}+P_2(U)_{\{i,k\},*})$;
its possible degree-zero entry is proportional to
$\sum_jU_{ij}U_{kj}=0$ for $i\ne k$.
Equation~\eqref{eq:leading-norm} forces the squared norm of the
parenthesis to be $4$. For orthogonal unit rows $u,v$ of $U$,
\[
 \sum_{j<l}(u_jv_l+u_lv_j)^2
       =1-2\sum_ju_j^2v_j^2\le1.
\]
For the corresponding rows, the triangle inequality gives
\[
 2=\|(U_2)_{\{i,k\},*}+P_2(U)_{\{i,k\},*}\|
   \le1+\|P_2(U)_{\{i,k\},*}\|\le2.
\]
Equality forces
$u_jv_j=0$ for every $j$. The rows of $U$ have disjoint supports,
so $U$ is a signed permutation. Relabel the $Y$ coordinates and
flip bits to make $U_1=I$; equality then gives $U_2=I$.
For $d=1$ the same conclusion follows directly by a bit flip.

Inductively suppose $U_j=I$ for $j<r$, with $r\ge3$.
By the convolution formula, the lowest-order terms in row $S$,
$|S|=r$, come from ordered disjoint decompositions $I\cup J=S$.
Two use an empty set; the other $2^r-2$ use only lower-degree
identity matrices. The coefficient of $\epsilon^r$ is therefore
$\theta_S(2(U_r)_{S,*}+(2^r-2)e_S)$.
Here $e_S$ is the unit row vector at $S$ in degree $r$.
Equation~\eqref{eq:leading-norm} makes its normalized norm $2^r$.
As both $(U_r)_{S,*}$ and $e_S$ are unit vectors, equality forces
$(U_r)_{S,*}=e_S$. Thus all $U_r=I$, and
\[
 B=s_0\bigotimes_{i=1}^d
       \frac12\begin{pmatrix}1+\theta_i&1-\theta_i\\1-\theta_i&1+\theta_i\end{pmatrix}.
\]
Absorb the scalar and set $\rho_i=(1-\theta_i)/(1+\theta_i)$.
\end{proof}

\subsection{Amplitudes and vertex weights}\label{sec:weights}

Normalization merges proportional rows into an Ising core, whereas
the theorem merges only identical rows. The information lost in
normalization is the amplitude and total weight above each core color.
Finite moments recover this information, proving the weighted and
unweighted results together.
For example, let $M=(vv^{\mathsf T})\otimes W(\rho)$, with
$v=(a,b)^{\mathsf T}$, positive $a\ne b$, and $\rho>0$, $\rho\ne1$.
The two amplitude colors above each bit state have proportional
but distinct rows. Diagonal normalization changes $vv^{\mathsf T}$
to $J_2$, merging these colors into one core state. The moment
argument keeps their weighted masses at amplitudes $a,b$ separate.

\begin{theorem}[One fixed positive vertex-weight vector]\label{thm:weighted}
Let $M=M^{\mathsf T}\in\mathbb R_{\ge0}^{q\times q}$ and
$w\in\mathbb R_{>0}^q$ have fixed algebraic entries.
For the identical-row classes $I_1,\ldots,I_s$ of $M$, define
\[
 R_{ab}=M_{ij}\quad(i\in I_a,\ j\in I_b),\qquad
 \bar w_a=\sum_{i\in I_a}w_i.
\]
For each support-component vertex set $S$ of $R$, put
$R_S=R[S]$ and $\bar w_S=\bar w|_S$.
Then $\PlGH{M,w}\in\mathsf{FP}$ if each $(R_S,\bar w_S)$, after
relabeling its colors and weights together, has one of the forms
\begin{align}\label{eq:weighted-forms}
 &(R_S,\bar w_S)=([0],[\lambda]),\qquad \lambda>0;\notag\\
 &R_S=(aa^{\mathsf T})\otimes T(\boldsymbol\rho),\qquad
     \bar w_S=\lambda\otimes\mathbf1_{2^d};\\
 &R_S=\begin{pmatrix}0&ab^{\mathsf T}\\ba^{\mathsf T}&0\end{pmatrix}
          \otimes T(\boldsymbol\rho),\qquad
     \bar w_S=\binom{\lambda}{\eta}\otimes\mathbf1_{2^d}.\notag
\end{align}
For each component, the parameters satisfy
\[
 a,\lambda\in\mathbb R_{>0}^{k},\qquad
 b,\eta\in\mathbb R_{>0}^{\ell},\qquad
 \rho_r>0,\quad\rho_r\ne1\quad(r\in[d]),
\]
with $k,\ell\ge1$ when present and $d\ge0$.
In every other case $\PlGH{M,w}$ is $\#\mathrm P$-hard under
polynomial-time Turing reductions.
\end{theorem}

\paragraph{Fields on the Ising core.}

\begin{lemma}\label{lem:core-fields}
Let $B=T(\boldsymbol\rho)$ have order $n=2^d$, with
$\rho_r>0$ and $\rho_r\ne1$ for every $r$.
For fixed algebraic $u,\mu,\nu\in\mathbb R_{>0}^{n}$,
the respective tractable cases are
\[
 \begin{aligned}
 \PlGH{B,u}:&\qquad u=\alpha\mathbf1_n
                 &&(\alpha>0),\\
 \PlGH{A_B,(\mu,\nu)}:&\qquad
     \mu=\alpha\mathbf1_n,\quad\nu=\beta\mathbf1_n
                 &&(\alpha,\beta>0).
 \end{aligned}
\]
Each problem is in $\mathsf{FP}$ under its displayed condition and
$\#\mathrm P$-hard otherwise.
\end{lemma}
\begin{proof}
The rows of $B$ are nonproportional. For the symmetric problem,
weight removal and square-root unaries (Lemma~\ref{lem:weights-remove}) supply, with unit background,
\[
 S=D^{1/2}BDBD^{1/2}\succ0,\qquad D=\operatorname{diag}(u).
\]
The positive matrix $Q=B^{\circ2}$ has a constant row sum $r$,
and $S_{ii}=u_i(Qu)_i$.
If $U=\max u_i>m=\min u_i$, then at a maximum of $u$ the diagonal
of $S$ is strictly greater than $Umr$, and at a minimum it is
strictly smaller. Theorem~\ref{thm:pd} makes $S$ hard.

For the bipartite problem, put
$\widetilde B=D_\mu^{1/2}BD_\nu^{1/2}$, where
$D_\mu=\operatorname{diag}(\mu)$ and $D_\nu=\operatorname{diag}(\nu)$.
The two positive definite Gram matrices
$S_X=\widetilde B\widetilde B^{\mathsf T}$ and
$S_Y=\widetilde B^{\mathsf T}\widetilde B$ are available with unit
background. Unless the source is hard, both have constant diagonal;
their equal traces make the constants equal, say to $h>0$. Thus
\[
 \mu_i(Q\nu)_i=h,\qquad \nu_j(Q^{\mathsf T}\mu)_j=h.
\]
If $\mu$ is nonconstant and $U=\max\mu_i$, positivity gives
$(Q^{\mathsf T}\mu)_j<rU$, hence $\nu_j>h/(rU)$ for every $j$.
At a maximum of $\mu$, the first equation is then impossible.
Thus $\mu$ is constant, and the second equation makes $\nu$
constant. Constant weights give the algorithms by multiplying
the Ising value by vertex factors, separately for the two
orientations in the bipartite case.
\end{proof}

\paragraph{Recovering amplitudes by moments.}

\begin{proof}[Necessity in Theorem~\ref{thm:weighted}]
Support interpolation and Proposition~\ref{prop:zero-one} leave
only positive symmetric components, positive rectangular bipartite
components, and zero singletons. Root restriction handles each
component separately.

For a positive component $M>0$, set $d_i=M_{ii}$, $a_i=\sqrt{d_i}$, and
$C_{ij}=M_{ij}/(a_i a_j)$. A source loop supplies the unary $d_i$;
unary interpolation supplies its powers. Attaching $d_i^{-1/2}$
at each endpoint of an $M$-edge realizes the normalization $C$.
Merge identical rows of $C$, obtaining a
unit-diagonal core $B$ indexed by classes $I_s$. Its rows are
nonproportional. With the additional unary $d_i^t$, its
background is
\begin{equation}\label{eq:symmetric-moments}
 u_s(t)=\sum_{i\in I_s}w_i d_i^t,\qquad t=0,1,\ldots.
\end{equation}
Weight removal and Lemma~\ref{lem:positive-core} make
$B\cong T(\boldsymbol\rho)$ unless the source is hard.
Lemma~\ref{lem:core-fields} requires $u_s(t)$ to be independent
of $s$ for every $t$.

Only finitely many moments are needed. List the distinct values
of the $d_i$ as $\delta_1,\ldots,\delta_k$ and put
$p_{sj}=\sum_{i\in I_s:d_i=\delta_j}w_i$. For
$0\le t<k$, equation~\eqref{eq:symmetric-moments} gives
\[
 \sum_{j=1}^k(p_{sj}-p_{s'j})\delta_j^t=0.
\]
The Vandermonde matrix is invertible, so $p_{sj}=\lambda_j$
is independent of $s$. Every $\delta_j$ occurs in some class,
where its mass is positive because $w>0$; hence $\lambda_j>0$.
In the original matrix, identical rows
are exactly the colors with the same normalized class $s$ and
the same amplitude $\sqrt{\delta_j}$. Consequently its actual
quotient is $(aa^{\mathsf T})\otimes T$ with
$a_j=\sqrt{\delta_j}$ and $\bar w=\lambda\otimes\mathbf1$.

For a bipartite component, write its positive cross block as $V$,
with background weights
$w_X,w_Y$. The two-edge return gadgets supply the positive unaries
\[
 a_i^2=\sum_j(w_Y)_jV_{ij}^2,\qquad
 b_j^2=\sum_i(w_X)_iV_{ij}^2.
\]
Their inverse square roots, attached to the respective endpoints,
normalize $V$ to $C_{ij}=V_{ij}/(a_i b_j)$. Merge its equal rows
and columns into classes $I_s,J_t$. Proportional rows of $V$
have proportional norms, so the resulting core $B$ has no
proportional rows or columns. Independent unary powers on the
prescribed input sides give
\begin{equation}\label{eq:bipartite-moments}
 \mu_s(m)=\sum_{i\in I_s}(w_X)_i a_i^{2m},\qquad
 \nu_t(n)=\sum_{j\in J_t}(w_Y)_j b_j^{2n}.
\end{equation}
After weight removal, Proposition~\ref{prop:rectangular} makes
$B=cT(\boldsymbol\rho)$ unless the source is hard.
Lemma~\ref{lem:core-fields} requires each vector in
\eqref{eq:bipartite-moments} to be constant. Apply the same
Vandermonde argument separately to the distinct left and right
amplitudes. Every core state has the same weighted mass
$\lambda_h$ at each left amplitude and $\eta_k$ at each right
amplitude. Merging actual twins therefore gives the last form
in \eqref{eq:weighted-forms}; absorb $c$ into an amplitude vector.
\end{proof}

For sufficiency, consider first a positive component and define
$F_a(k)=\sum_i\lambda_i a_i^k$.
Its value on any graph is
\begin{equation}\label{eq:weighted-algorithm}
 Z_{R,\bar w}(G)=\prod_v F_a(\deg v)
                         \prod_r Z_{W(\rho_r)}(G).
\end{equation}
For a connected bipartite input with parts $L,R$, the bipartite
target form instead contributes
\[
 \left(\prod_{v\in L}F_a(\deg v)\prod_{v\in R}F_b(\deg v)
      +\prod_{v\in R}F_a(\deg v)\prod_{v\in L}F_b(\deg v)\right)
      \prod_r Z_{W(\rho_r)}(G),
\]
where $F_b(k)=\sum_j\eta_j b_j^k$; it vanishes on a nonbipartite
input. For an isolated vertex the answer is simply the total
vertex weight. A zero singleton contributes only to edgeless
inputs. Finally, sum over target components and multiply over
input components. These formulas and the FKT algorithm
prove sufficiency and polynomial bit complexity.

To finish Theorem~\ref{thm:main}, apply the necessity argument with
original weights $w_i=1$. The masses recovered by the moment
equations are then integer multiplicities, equal above every
core state. Repeat each amplitude that many times in $a$ or $b$;
this recovers precisely the original matrix forms
\eqref{eq:tractable-forms}. Formula~\eqref{eq:weighted-algorithm}
with unit amplitude weights, and its bipartite counterpart, prove
sufficiency.

Nonnegative vertex weights require no new argument: first delete
zero-weight colors and then apply Theorem~\ref{thm:weighted}.
If none remain, the value is zero on every nonempty input and one
on the empty graph.

\section{Real parameters and applications}\label{sec:applications}

The real-parameter extension in Theorem~\ref{thm:real} lets us apply
the classifications directly to Boltzmann weights.

\subsection{Arbitrary fixed real weights}\label{sec:real}

We retain the exact model of \cite[Appendix~A]{cmy2026}:
the weights generate a fixed field
\[
 F=\mathbb Q(X_1,\ldots,X_b)(\beta),\qquad
 z=\sum_{j=0}^{e-1}z_j(X)\beta^j,\quad z_j\in\mathbb Q(X),
\]
Here $X=(X_1,\ldots,X_b)$ consists of algebraically independent
real constants, and $\beta$ has degree $e$ over $\mathbb Q(X)$.
An exact answer $z\in F$ is encoded by the rational functions $z_j$.
The minimal polynomial of $\beta$, the expressions for the source
weights, and compatible presentations of auxiliary fields are fixed
data. Equality is checked symbolically in these representations.

\begin{theorem}[Fixed real parameters]\label{thm:real}
Let
\[
 M=M^{\mathsf T}\in\mathbb R_{\ge0}^{q\times q},\qquad
 w\in\mathbb R_{>0}^q
\]
be fixed in the field representation above.
The classifications and hardness alternatives of
Theorems~\ref{thm:main} and~\ref{thm:weighted} hold for
$\PlGH{M}$ and $\PlGH{M,w}$, respectively.
For $w\ge0$, the deletion rule for zero-weight colors applies.
\end{theorem}
\begin{proof}
\emph{Exact reductions.}
The fixed numbers of variables and basis elements make arithmetic
and equality polynomial-time on rational-function encodings of
polynomial degree and coefficient bit length. Products of
polynomially many fixed weights have these bounds, as do
polynomial-size interpolation systems by fraction-free elimination.
Grouping colorings by their fixed-weight occurrence counts bounds
the exact output size in the same way. Thus product interpolation,
root restriction, and weight removal remain valid. Square roots
and spectral gadgets use fixed finite algebraic extensions.

Suppose an answer computed in an auxiliary field $F$ lies in the
original field $F_0\subseteq F$. Choose algebraically independent
elements $Y$ over $F_0$ so that $F/F_0(Y)$ is finite algebraic.
Normalized trace to $F_0(Y)$ gives $z=p(Y)/q(Y)$.
Since $z\in F_0$, coefficient comparison in $p=zq$ returns
\[
 z=\frac{[Y^\alpha]p}{[Y^\alpha]q}
 \quad\text{for any nonzero denominator coefficient}.
\]
Here $[Y^\alpha]$ extracts the coefficient of the monomial $Y^\alpha$.
Fixed compatible presentations keep all these operations within
polynomial degree and bit bounds.

\emph{Transfer of the unweighted dichotomy.}
We transfer hardness through nearby rational matrices that preserve
all multiplicative relations.
List the positive entries of $M$ as $m_1,\ldots,m_s$ and consider
their complete relation lattice
\[
 \Lambda=\left\{z\in\mathbb Z^s:\prod_{\alpha=1}^s m_\alpha^{z_\alpha}=1\right\}.
\]
The vector $\ell=(\log m_1,\ldots,\log m_s)$ lies in
$\Lambda^\perp$, a subspace defined by rational linear equations.
Choose integer vectors $u_1,\ldots,u_h$ forming a real basis of
$\Lambda^\perp$, and write $\ell=\sum_jt_j u_j$. Replacing each $e^{t_j}$ by a
nearby positive rational $r_j$ produces rational entries
$m'_\alpha=\prod_jr_j^{(u_j)_\alpha}$ arbitrarily close to $m$.
They preserve every multiplicative relation, hence all equal-length
product relations. Equal entries remain equal; zero positions are
kept. Product interpolation gives
$\PlGH{M'}\le_T\PlGH{M}$ exactly.
Here $M$ is fixed. The lattice, basis, and one suitable rational
replacement are constants of its reduction; the input is the graph,
not a description of an arbitrary real matrix from which these
choices must be computed uniformly.

The structural class \eqref{eq:tractable-forms} is closed among
matrices with a fixed support. Indeed, there are only finitely many
factor sizes and relabelings. For a positive block in fixed coordinates
$(i,x)$, where $i$ indexes amplitudes and $x$ is a bit string, recover
\[
 a_i=\sqrt{M_{(i,0),(i,0)}},\qquad
 \rho_r=\frac{M_{(1,0),(1,e_r)}}{M_{(1,0),(1,0)}}.
\]
These are continuous on the fixed positive support. For a bipartite
block, fix one left amplitude to $1$ and recover the other amplitudes
and ratios from a positive reference row and column. In both cases
all remaining entries satisfy continuous tensor identities.
Limits may have $\rho_r=1$, already allowed in the
unweighted theorem. If $M$ is outside the class, choose a sufficiently
close rational $M'$ outside it. The algebraic theorem makes $M'$
hard and the preceding reduction makes $M$ hard. Inside the class,
the same amplitude and Ising formulas work over the fixed field.

\emph{Positive vertex weights.}
The support and weight-removal reductions just established give
the real versions of the normalized-core statements directly from
the now-proved unweighted dichotomy. A positive unit-diagonal
core with distinct rows must be a nondegenerate Ising tensor;
a positive rectangular core with nonproportional rows and columns
must be a square Ising tensor after independent relabelings.
The same dichotomy excludes a positive definite tractable matrix
with nonconstant diagonal. It therefore proves
Lemma~\ref{lem:core-fields} over the real field. The normalization,
finite Vandermonde equations, and algorithms above are unchanged,
proving the full weighted statement.
\end{proof}

For signed real symmetric $M$, let $P(M)_{ij}=\mathbf1_{M_{ij}\ne0}$;
write $|M|$ and $\operatorname{sgn}M$ for entrywise absolute values
and signs, with $\operatorname{sgn}(0)=0$. Equal entry products preserve
their magnitudes and signs, so product interpolation gives, for
every fixed $w>0$,
\begin{gather*}
 \PlGH{P(M),w}\le_T\PlGH{|M|,w}\le_T\PlGH{M,w},\\
 \PlGH{\operatorname{sgn}M,w}\le_T\PlGH{M,w}.
\end{gather*}
Only the signs of the finitely many fixed entries are needed.
Consequently a nonbasic support component, or a magnitude problem
outside the weighted tractable class, makes the signed problem hard.
In particular $|M|$ outside \eqref{eq:tractable-forms} implies
unweighted hardness. The zero--one criterion of
Proposition~\ref{prop:zero-one} therefore holds with every fixed
positive real vertex-weight vector.

\subsection{Targets and input graph classes}\label{sec:bipartite-inputs}

For unweighted target graphs, Proposition~\ref{prop:zero-one}
recovers exactly the unrestricted-input tractable list of Dyer
and Greenhill \cite{dyerGreenhill2000,dyerGreenhill2004}:
fully looped cliques, loopless complete bipartite graphs, and
loopless isolated vertices. Every other target remains
$\#\mathrm P$-complete on planar inputs. Thus the additional
planar tractability is a weighted phenomenon.

A complete dichotomy for planar bipartite input graphs follows by
doubling the state set; its tractable class can be strictly larger.
Define
\[
 D_M=\begin{pmatrix}0&M\\M&0\end{pmatrix},\qquad
 \widetilde w=(w,w).
\]
On a bipartite graph $G$,
$Z_{D_M,\widetilde w}(G)=2^{c(G)}Z_{M,w}(G)$, where $c(G)$
counts connected components, including isolated vertices; on a
nonbipartite component the left side vanishes. Consequently the
classification on planar bipartite inputs is obtained by applying
Theorem~\ref{thm:weighted} to $(D_M,\widetilde w)$.
More generally, a rectangular interaction $V\ge0$ on prescribed
sublattices, with respective weights $\mu,\nu>0$, has the
classification of $(A_V,(\mu,\nu))$, by the root restriction
argument in Section~\ref{sec:nonnegative}.

For a concrete strict enlargement, let $0<r<1$, put
$T=W(r)\otimes W(r)$, and let $P$ swap its Boolean coordinates.
Then $M=PT=TP$ is symmetric and positive. Relabeling one side of
a bipartite input changes $M$ to $T$, giving an algorithm.
On general planar inputs it is hard: $M$ has full rank and
$\det M=-(1-r^2)^4<0$. Indeed, full rank excludes an amplitude
factor of dimension greater than one, so a positive tractable
four-state matrix must be $cW(\rho_1)\otimes W(\rho_2)$, with
$c>0$ and $\rho_1,\rho_2\ne1$. Its determinant is
\[
 c^4(1-\rho_1^2)^2(1-\rho_2^2)^2>0.
\]

Theorems~\ref{thm:main} and~\ref{thm:weighted} also hold for graphs of
any fixed bounded orientable genus, supplied with an embedding. The amplitude
and component decompositions work on every graph; the genus-$h$
extension of the FKT algorithm uses at most $4^h$ Pfaffians \cite[Theorem~2.3]{pham2020}.
This includes antiferromagnetic parameters: the even-subgraph
polynomial and its signed Pfaffian expression agree as polynomials,
so their identity also holds when $(1-\rho)/(1+\rho)<0$.
For fixed genus the computation has polynomial bit complexity,
and the hard cases already contain the planar subclass.

\subsection{Spin systems and imaginary-time kernels}

Here we use the entrywise logarithm
$\log_\circ M=(\log M_{ij})_{i,j}$. For an Ising tensor, centering
removes the scalar factor, and the resulting rank counts its
nondegenerate coordinates.

\begin{corollary}\label{cor:log-rank}
Let $M=M^{\mathsf T}\in\mathbb R_{>0}^{q\times q}$ have distinct
rows and $M_{ii}=\delta$ for all $i$, and let $w\in\mathbb R_{>0}^q$.
In the tractable cases of Theorem~\ref{thm:weighted},
\[
 q=2^d,\qquad M\cong\gamma T(\boldsymbol\rho),\qquad
 w=\lambda\mathbf1_q,
\]
where $d\ge0$, $\gamma,\lambda>0$, and $\rho_r>0$, $\rho_r\ne1$.
Consequently, with $P=I_q-J_q/q$,
\[
 \operatorname{rank}\!\bigl(P\log_\circ(M)P\bigr)=d.
\]
\end{corollary}
\begin{proof}
In $(aa^{\mathsf T})\otimes T$, constant diagonal makes $a$
constant; distinct rows then make it a singleton. The weighted
criterion gives constant $w$. For the column vectors
$\chi_j=((-1)^{x_j})_{x\in\{0,1\}^d}$,
\[
 P\log_\circ(M)P=-\frac12\sum_j(\log\rho_j)\chi_j\chi_j^{\mathsf T}.
\]
The $\chi_j$ are orthogonal and all coefficients are nonzero,
giving rank $d$.
\end{proof}

\begin{corollary}[Clock models]\label{cor:clock}
Fix $q\ge2$, $K\in\mathbb R$, and $w\in\mathbb R_{>0}^q$. For
$i,j\in\{0,\ldots,q-1\}$, put
\[
 C(q,K)_{ij}=\exp\!\left(K\cos\frac{2\pi(i-j)}q\right).
\]
Then $\PlGH{C(q,K),w}\in\mathsf{FP}$ if
\[
 K=0
 \quad\text{or}\quad
 \bigl(q\in\{2,4\}\ \text{and}\ w=\lambda\mathbf1_q
       \text{ for some }\lambda>0\bigr),
\]
and $\PlGH{C(q,K),w}$ is $\#\mathrm P$-hard otherwise.
\end{corollary}
\begin{proof}
For $K\ne0$, the diagonal is constant and each row has its unique
maximum there if $K>0$, or unique minimum if $K<0$; thus the rows
are distinct. For $q\ge3$, set $c_i=\cos(2\pi i/q)$ and
$s_i=\sin(2\pi i/q)$. Then
\[
 \log_\circ C(q,K)=K(cc^{\mathsf T}+ss^{\mathsf T})
\]
has rank two and is already centered. Corollary~\ref{cor:log-rank}
forces $q=4$ and constant weights. The surviving cases factor as
\[
 C(2,K)=e^KW(e^{-2K}),\qquad
 C(4,K)\cong e^KW(e^{-K})\otimes W(e^{-K}),
\]
using the cyclic Gray ordering $00,01,11,10$ for $q=4$.
The same corollary excludes nonconstant weights at $q=2$.
If $K=0$, the answer is $(\sum_iw_i)^{|V(G)|}$.
\end{proof}

\begin{corollary}[Coupled Ising interactions]\label{cor:characters}
Fix $n$, a set $\mathcal A\subseteq\mathbb F_2^n\setminus\{0\}$,
and couplings $J_a\in\mathbb R\setminus\{0\}$ for $a\in\mathcal A$.
For $x,y\in\mathbb F_2^n$, define
\[
 M_{xy}=\exp\!\left(\sum_{a\in\mathcal A}
                 J_a(-1)^{a\cdot x}(-1)^{a\cdot y}\right),
\]
where dot products are taken modulo two.
For fixed $w\in\mathbb R_{>0}^{2^n}$, put
\[
 F(x)=(a\cdot x)_{a\in\mathcal A},\qquad
 \bar w_z=\sum_{x:F(x)=z}w_x\quad(z\in\operatorname{im}F).
\]
Then $\PlGH{M,w}\in\mathsf{FP}$ if
\[
 \operatorname{rank}_{\mathbb F_2}\mathcal A=|\mathcal A|,
 \qquad \bar w_z=\lambda\quad(z\in\operatorname{im}F)
 \quad\text{for some }\lambda>0,
\]
and $\PlGH{M,w}$ is $\#\mathrm P$-hard otherwise.
For $w=\mathbf1$, the fiber sums are automatically constant, so the
criterion is simply linear independence of $\mathcal A$.
\end{corollary}
\begin{proof}
Let $r=\operatorname{rank}_{\mathbb F_2}\mathcal A$ and
$m=|\mathcal A|$. Independence of distinct characters shows that
two rows agree exactly when their $F$-values agree. The quotient
has $2^r$ distinct rows and constant diagonal. Its centered
entrywise logarithm, indexed by $z\in\operatorname{im}F$, is
$\sum_{a\in\mathcal A}J_a\chi_a\chi_a^{\mathsf T}$, of rank
$m$, where $\chi_a(z)=(-1)^{z_a}$. Corollary~\ref{cor:log-rank} therefore requires $m=r$ and
constant $\bar w$. Conversely, if $m=r$, the quotient is
\[
 e^{\sum_aJ_a}\bigotimes_{a\in\mathcal A}W(e^{-2J_a}),
\]
so constant fiber sums give the algorithm. Unit weights have fiber
sum $2^{n-r}$. The empty set $\mathcal A$ gives a one-state quotient.
\end{proof}

For an $n$-layer Ashkin--Teller system \cite{grestWidom1981},
ordinary layer couplings have labels $e_i$, and four-spin couplings
between layers $i,j$ have labels $e_i+e_j$. If all ordinary couplings
are nonzero, any nonzero four-spin coupling creates a dependence
and gives hardness. With only four-spin couplings, tractability
with unit weights is equivalent to the graph of coupled layers
being a forest: a cycle sums to zero, while leaf removal proves
independence for a forest. This fixed graph of layers is distinct
from the variable planar input graph.

For real single-site potentials $U_i$ and inverse temperature
$\beta\ne0$, set $w_i=e^{-\beta U_i}$. An identical-row class $I$
has effective potential
\[
 U_I^{\mathrm{eff}}=-\beta^{-1}\log\sum_{i\in I}e^{-\beta U_i}.
\]
The weighted theorem requires this potential to be independent of
the Ising coordinates within each amplitude and each bipartite side.
Thus unequal microscopic potentials can preserve tractability
when their aggregated weights agree.

As a short example, the spin-one Blume--Capel interaction
\[
 M_{st}=e^{Kst},\qquad w_s=e^{-Ds^2+Fs},\qquad s,t\in\{-1,0,1\},
\]
is hard for every fixed real $K\ne0$ and all finite real $D,F$.
Indeed, with $x=e^K$, its determinant is
$(x-1)^3(x+1)/x^2\ne0$. A positive full-rank tractable matrix
has order a power of two, excluding order three regardless of
the positive vertex weights. At $K=0$ the interaction is all ones.

For a fixed local Hamiltonian $h$, imaginary-time evolution gives
the ordinary matrix exponential $K_\tau=e^{-\tau h}$. We classify
contraction of this fixed kernel over planar graphs, namely
$Z_{K_\tau}(G)$ with unit vertex weights.

\begin{corollary}[Imaginary-time kernels]\label{cor:quantum-kernel}
Let $h$ be a fixed real symmetric matrix of order $q$ with
$h_{ij}\le0$ for $i\ne j$ (a \emph{stoquastic} Hamiltonian),
whose graph of pairs $i\ne j$ with $h_{ij}<0$ is connected.
For each fixed $\tau>0$, $\PlGH{e^{-\tau h}}$ is in $\mathsf{FP}$ if,
after a color relabeling,
\[
 q=2^d,\qquad h=cI-\sum_{r=1}^d b_rX_r,\qquad
 c\in\mathbb R,\quad b_r>0,
\]
where $X_r$ flips bit $r$ and $d=0$ is allowed.
Otherwise it is $\#\mathrm P$-hard, with kernel entries represented
as in Theorem~\ref{thm:real}.
\end{corollary}
\begin{proof}
Shifting $-h$ by a scalar and using the exponential series gives
$K_\tau>0$; its eigenvalues are positive, so $K_\tau\succ0$.
Theorems~\ref{thm:pd} and~\ref{thm:real} require
$K_\tau\cong\gamma T(\boldsymbol\rho)$, $0<\rho_r<1$, in the
tractable case. Taking its unique real symmetric logarithm and
using \eqref{eq:ising-log} gives the stated form of $h$.
Conversely,
\[
 e^{-\tau h}=e^{-\tau c}\prod_r\cosh(\tau b_r)\,
                 \bigotimes_r W(\tanh(\tau b_r)),
\]
so the FKT algorithm applies.
\end{proof}
The criterion is independent of $\tau$; disconnected hopping
graphs are handled componentwise.

For the quantum clock hopping term $h_q=-(S_q+S_q^{\mathsf T})$
\cite[Appendix~A]{vishnu2025}, let $S_q$ cyclically shift the
$q\ge2$ basis states. The corollary gives tractability exactly at
$q=2,4$, for every fixed $\tau>0$. For $q\ge3$ its hopping graph
is the cycle $C_q$, a cube only at $q=4$: the Gray ordering above gives
$h_4=-X_1-X_2$, while $h_2=-2X_1$.

\section*{Acknowledgments}
We thank Jin-Yi Cai and Zhuxiao Tang for informing us that they had
independently obtained the main result and for explaining its
significance.

The authors used OpenAI's ChatGPT in preparing
this manuscript, including for language editing and \LaTeX{} preparation.
ChatGPT also contributed to the exploratory development of the main
technical arguments, in particular entropy continuation and the
recovery of product structure from logarithmic support. All claims
and proofs were independently checked and finalized by the authors,
who take full responsibility for the content.

\end{document}